\documentclass{article}

\usepackage{amsmath}

\usepackage{amssymb}

\usepackage{amsthm}

\usepackage[toc,page]{appendix}

\usepackage[mathscr]{eucal}

\newcommand{\M}{\mathcal{M}}

\newcommand{\N}{\mathcal{N}}

\newcommand{\X}{\mathfrak{X}}

\newcommand{\h}{\bar{h}}

\newcommand{\D}{\mathfrak{D}}

\newcommand{\I}{\mathfrak{I}}

\newcommand{\Div}{\mathrm{div}_p\,}

\newcommand{\Ric}{\mathrm{Ric}}

\newcommand{\R}{\mathbb{R}}

\newcommand{\vol}{\,dx\,d\mu}

\theoremstyle{plain}

\newtheorem{Theorem}{Theorem}

\newtheorem{assumption}{Assumption}

\newtheorem{Corollary}{Corollary}

\newtheorem{Proposition}{Proposition}

\newtheorem{Definition}{Definition}

\newtheorem{Lemma}{Lemma}

\theoremstyle{remark}

\newtheorem*{Remark}{Remark}

\title{Exponential convergence to equilibrium\\ for kinetic Fokker-Planck equations}

\author{Simone Calogero\\[0.5cm]Departamento de Matem\'atica Aplicada\\[0.2cm] Universidad de Granada, Spain\\[0.2cm] E-Mail: calogero@ugr.es}

\date { }

\begin{document}

\maketitle

%

\begin{abstract}
A class of linear kinetic Fokker-Planck equations with a non-trivial diffusion matrix 
and with periodic boundary conditions in the spatial variable is considered. After formulating the problem in a geometric setting, the question of the rate of convergence to equilibrium is studied within the formalism of differential calculus on Riemannian manifolds. Under explicit geometric assumptions on the velocity field, the energy function and the diffusion matrix, it is shown that global regular solutions converge in time to equilibrium with exponential rate. The result is proved by estimating the time derivative of a modified entropy functional, as recently proposed by Villani. For spatially homogeneous solutions the assumptions of the main theorem reduce to the curvature bound condition for the validity of logarithmic Sobolev inequalities discovered by Bakry and Emery. The result applies to the relativistic Fokker-Planck equation in the low temperature regime, for which exponential trend to equilibrium was previously unknown.
\end{abstract}


\section{Introduction}
\label{intro}
The evolution of many physical or biological systems is characterized by two kinds of driving mechanism: diffusion and friction. 
The competition between these two types of dynamics may lead the system to a thermodynamical equilibrium. The purpose of this paper is to study the rate of convergence to equilibrium for a class of linear models that exhibit this kind of behavior. 
The simplest model in this class is the Fokker-Planck equation~\cite{Risken} on the density function of the system:
\begin{equation}\label{FPeq}
\partial_t\rho=\Delta\rho +\nabla\cdot(\xi\rho)\:,\qquad t>0\:,\quad \xi\in\R^N\:.
\end{equation}
In this model the density function $\rho$ depends on the variables $(t,\xi)$, the diffusion term is given by $\Delta\rho$ and the friction term by $\nabla\cdot (\xi\rho)$. All sufficiently regular solutions of~\eqref{FPeq} converge in time to a Maxwellian type distribution, with exponential rate of convergence~\cite{Carrillo}. The convergence holds for instance in the $L^1$-norm.

An important generalization of~\eqref{FPeq}, often considered in the mathematical and in the physical literature~\cite{AMTU, Risken}, is
\begin{equation}\label{FPeq2}
\partial_t\rho=\nabla\cdot(D(\nabla\rho +\rho\nabla E))\:,
\end{equation}
where $D=D(\xi)$ is the diffusion matrix and $E=E(\xi)$ is the energy function ($E=|\xi|^2/2$ for~\eqref{FPeq}). It is assumed that $D$ is positive definite and that  
\[
\Theta^{-1}=\int_{\R^N} e^{-E}\,d\xi
\]
is bounded. Equation~\eqref{FPeq2} admits an unique invariant probability measure, which is given by $d\mu=\rho_\infty(\xi)\,d\xi$, where $\rho_\infty=\Theta\,e^{-E}$. Non-negative solutions of~\eqref{FPeq2} with unit mass
converge to the equilibrium state $\rho_\infty$ with exponential rate if the matrix $D$ and the function $E$ satisfy an inequality known as the curvature bound condition~\cite{BE,B}. Let us briefly recall the argument of the proof. In terms of $h(t,\xi)=\rho(t,\xi)/\rho_\infty(\xi)$, equation~\eqref{FPeq2} takes the form
\[
\partial_t h=\nabla\cdot(D\nabla h)- D\nabla E\cdot\nabla h\:, 
\]
or, equivalently, 
\begin{equation}\label{FPeq3}
\partial_th=\Delta^Gh+Qh\:,
\end{equation}
where $\Delta^G$ denotes the Laplace-Beltrami operator
associated with the Riemannian metric $G=D^{-1}$ and $Q$ is the vector field
\[
Qh=D\nabla\log u\cdot\nabla h\:,\quad u=\sqrt{\det D}\,e^{-E}\:.
\]   
The entropy functional associated with~\eqref{FPeq3} is given by
\begin{equation}\label{entropyfun}
\D[h]=\int_{\R^N} h\log h\,d\mu\:,
\end{equation}
and satisfies 
\begin{equation}\label{entropyid}
\frac{d}{dt}\D[h]=-\I[h]\:,
\end{equation}
where 
\[
\I[h]=\int_{\R^N}\frac{D\nabla h\cdot\nabla h}{h}\,d\mu
\]
is the entropy dissipation functional. Let $\Ric^G$ and $\nabla^G$ denote the Ricci curvature and the Levi-Civita connection of $G$. Bakry and Emery proved in~\cite{BE,B} that if the curvature bound condition
\begin{equation}\label{curvaturebound}
\Ric^G-\nabla^GQ\geq \alpha\,G\:,\quad\text{for some }\ \alpha>0\:,
\end{equation}
is satisfied, then the logarithmic Sobolev inequality $\D[f]\leq (2\alpha)^{-1}\I[f]$ holds for all sufficiently regular probability densities $f$. Replacing in~\eqref{entropyid} we obtain 
\[
\frac{d}{dt}\D[h]\leq -2\alpha\,\D[h]\:,
\]
whence the entropy functional decays exponentially as $O(e^{-2\alpha t})$. The classical Csisz\'ar-Kullback inequality~\cite{CS},
\begin{equation}\label{CKineq}
\| h-1\|_{L^1(d\mu)}\leq \sqrt{2\D}\:,
\end{equation}
implies that $h$ converges to $1$  as $t\to\infty$ in $L^1(d\mu)$ with exponential rate or, equivalently, the solution $\rho$ of~\eqref{FPeq2} converges to $\rho_\infty$ in $L^1(d\xi)$ with exponential rate.

Fokker-Planck type equations appear also in kinetic theory and these are the subject of the present investigation. Assuming periodic boundary conditions in space, the simplest kinetic Fokker-Planck (or Kramers) equation is given by~\cite{Risken}
\begin{equation}\label{kineticFP}
\partial_t f+ p\cdot\nabla_x f=\Delta_p f+\nabla_p\cdot(pf)\:,\quad t>0\:,\ x\in\mathbb{T}^N\:,\ p\in\R^N\:.
\end{equation}
Here $f=f(t,x,p)\geq 0$ is the particles distribution in phase-space, with $(t,x)$ denoting the space-time variables and $p$ the momentum variable. Equation (\ref{kineticFP}) describes the kinetic motion of a system of particles undergoing stochastic collisions with the molecules of a homogeneous fluid in thermal equilibrium. From a mathematical point of view,~\eqref{kineticFP} is more complicated than~\eqref{FPeq}, due to the presence of the transport derivative $p\cdot\nabla_x$ and the fact that the diffusion operator $\Delta_p$ is degenerate, i.e., it acts only on the momentum variable. Moreover equation~\eqref{FPeq} can be seen as the spatially homogeneous version of~\eqref{kineticFP}.

The problem of how fast the solutions of~\eqref{kineticFP} converge to equilibrium was solved only recently by H\'erau and Nier~\cite{Herau} and by Villani~\cite{hypo}. Both references establish exponential convergence to equilibrium, however by completely different methods. In~\cite{Herau} the problem is tackled by spectral analysis techniques for hypoellitpic operators---exponential rate of convergence is implied by the existence of a spectral gap in the spectrum of the Fokker-Planck operator---, while the proof given in~\cite[Th.~28]{hypo} is based on the study of the evolution of a properly modified entropy functional. (See~\cite{DV} for an earlier study of the trend to equilibrium for~\eqref{kineticFP}. In the latter reference the authors prove that convergence to equilibrium occurs as fast as $O(t^{-1/\varepsilon})$, for any $\varepsilon>0$.)

In this paper the entropy method is applied to study the trend to equilibrium for the following generalization of~\eqref{kineticFP}:
\begin{equation}\label{kineticFP2}
\partial_t f+v(p)\cdot\nabla_x f=\nabla_p(D(\nabla_p f+ f\nabla_p E))\:,
\end{equation}
where $D=D(p)$ is the diffusion matrix and $E=E(p)$ is the energy function. Equation~\eqref{kineticFP2} reduces to~\eqref{FPeq2} in the spatially homogeneous case. The vector field $v$ in the transport term is the velocity field;  the two most important examples for the applications are the classical velocity field
\[
v(p) = p
\]
and the relativistic velocity field
\[
 v(p)=\frac{p}{\sqrt{1+|p|^2}}\:.
\]
Setting $h=f/e^{-E}$, we may rewrite~\eqref{kineticFP2} in the form
\begin{equation}\label{kineticFP3}
\partial_t h+v(p)\cdot\nabla_xh=\Delta_p h+Wh\:,\quad t>0\:,\ x\in\mathbb{T}^N\:,\ p\in\R^N\:,
\end{equation}
where $\Delta_p$ is the Laplace-Beltrami operator associated with the metric $g=D^{-1}$ 
and 
\[
Wh=D\nabla_p\log u\cdot\nabla_p h\:,\quad u=\sqrt{\det D}\,e^{-E}\:.
\]
The main result of this paper is presented in Section~\ref{mainresult}. It is proved that under suitable assumptions on the functions $v,D,E$ , which take the form of geometric inequalities involving $v, g$ and $W$, smooth solutions of~\eqref{kineticFP3} with unit mass converge in time to the equilibrium state $h_\infty\equiv 1$ with exponential rate of convergence. The proof is based on a generalization of the argument for spatially homogeneous solutions outlined before, the main difference being that the entropy functional~\eqref{entropyfun} is replaced by a properly modified entropy functional, as proposed recently by Villani~\cite{hypo}. Specifically, the modified entropy used in the present paper is given by
\[
\mathcal{E}[h]=k\D[h]+a\I_{pp}[h]+2b\I_{xp}[h]+c\I_{xx}[h]\:,
\]
where $a,b,c,k$ are properly chosen positive constants,
\[
\D[h]=\int_{\mathbb{T}^N}\int_{\R^N}h\log h\,d\mu\,dx
\]
is the classical entropy, see~\eqref{entropyfun},
\[
\I_{pp}[h]=-\frac{d}{dt}\D[h]=\int_{\mathbb{T}^N}\int_{\R^N}g(\partial_ph,\partial_p\log h)\,d\mu\,dx
\]
is the entropy dissipation functional (or Fisher information), and
\begin{align*}
&\I_{xx}[h]=\int_{\mathbb{T}^N}\int_{\R^N}g(\mathscr{A}_xh,\mathscr{A}_x\log h)\,d\mu\,dx\:,\\
&\I_{xp}[h]=\int_{\mathbb{T}^N}\int_{\R^N}g(\mathscr{A}_xh,\partial_p\log h)\,d\mu\,dx\:,
\end{align*}
where $\mathscr{A}_xh$ is a modified gradient of $h$ in the $x$-variable ($\mathscr{A}_xh=\partial_x h$ when $v(p)=p$); see Definition~\ref{modifiedgradient} in Section~\ref{mainresult}. Note that by Young's inequality, the constant $a,b,c$ can be chosen so that $\mathcal{E}[h]\geq k\mathfrak{D}[h]$, whence exponential decay of the modified entropy entails exponential decay of the classical entropy. The reason for choosing the modified entropy in the above form is based on the following identities, which are proved in Section~\ref{derivative}. Firstly, $d\D[h]/dt=-\I_{pp}[h]$; secondly, the mixed derivatives term $\I_{xp}[h]$ appears in the time derivative of $\I_{pp}[h]$:
\[
\frac{d}{dt}\I_{pp}[h]=-2\I_{xp}[h]+\dots
\]
where the omitted terms here and below are geometric functions of $v,g,W$ and the derivatives of $h$. Next, differentiating in time the function $\I_{xp}[h]$, the positive term $\I_{xx}[h]$ shows up:
\[
\frac{d}{dt}\I_{xp}[h]=-\I_{xx}[h]+\dots
\]
while the time derivative of $\I_{xx}[h]$ does not reproduce any of the other functionals, i.e., 
\[
\frac{d}{dt}\I_{xx}[h]=\dots
\]
The above identities lead to a differential inequality on the modified entropy of the form
\begin{equation}\label{E}
\frac{d}{dt}\mathcal{E}[h]\lesssim-\mathcal{E}[h]+\dots
\end{equation}
provided we assume that the following logarithmic Sobolev inequality in the variables $(x,p)$ holds:  
\[
\D[h]\lesssim\I_{xx}[h]+\I_{pp}[h]\:.
\]
In Appendix~\ref{logproof} we show that the latter bound is satisfied when properly defined metric and vector field on the  product manifold $\mathbb{T}^N\times\R^N$ verify the Bakry-Emery curvature bound condition~\eqref{curvaturebound}. 
The remaining assumptions of the main theorem ensure that the omitted terms in the inequality~\eqref{E} give a non-positive contribution and Gr\"onwall's inequality completes the proof. The full list of assumptions is given in Section~\ref{mainresult}. They take the form of geometric inequalities on the functions $g,v,W$, which are trivially satisfied for the classical Fokker-Planck equation~\eqref{kineticFP}.

It should be noticed that the trend to equilibrium for~\eqref{kineticFP2} is also studied in~\cite{hypo}.
The main differences between Villani's approach and the one in this paper are the following. In~\cite{hypo} the author exploits the fact that~\eqref{kineticFP2} can be written in the form
\[
\partial_th=(A^*A+B)h\:,
\]
where $B=-v(p)\cdot\nabla_x$, $A^*$ is the adjoint of $A$ in the Hilbert space $L^2(d\mu)$ and $A=\sigma\nabla$, $\sigma=\sqrt{D}$. The proof of~\cite[Th.~28]{hypo} makes crucial use of the iterated commutators 
\[
[A,B]\:,\quad [B,[A,B]]\:,\quad [B,[B,[A,B]]]\:,\quad\dots
\] 
in the spirit of H\"ormander's hypoellipticity theory~\cite{H}. However the use of commutators, while natural in the context of regularity theory, presents some disadvantages for the problem of convergence to equilibrium. In particular it leads to very heavy and sometimes obscure calculations, which, as pointed out by Villani at the beginning of the proof of Lemma 32 in~\cite{hypo}, ``might be an indication that a more appropriate formalism is still to be found". The main theorem of the present paper is proved using the formalism of differential geometry, which helps to clarify the meaning of many long expressions that have to be controlled in Villani's work. Moreover, given a diffusion matrix $D$, a velocity field $v$ and an energy function $E$, one may check directly if our assumptions are verified. As opposed to this, one has to find a suitable way to decompose the vector fields $A,B$ and the iterated commutators in order to verify the assumptions of~\cite[Th.~28]{hypo}. Apart from the different approach to the problem, several ideas introduced by Villani in his important work will be adapted (and simplified) to the present context, resulting in a less technical proof. %

In Section~\ref{appl} we apply the main result of the paper to establish exponential convergence to equilibrium for the relativistic Fokker-Planck equation~\cite{AC0, DH2, H1} when the temperature of the surrounding bath is sufficiently small. We do not know whether the small temperature assumption is only a technical condition or a necessary physical one. However in~\cite{AC1}, see also~\cite{angst}, it is shown that spatially homogeneous solutions converge exponential fast to equilibrium for all temperatures, although the result is proved not in $L^1$ but in a weighted $L^2$ norm. 
  
Finally, we remark that a natural and interesting extension of the results of this paper would be to consider the equation on the whole space, $x\in\R^N$, with an external confining potential (this is the case studied in~\cite{Herau, hypo}). 

\section{Set-up}\label{setup}
This section is devoted to introduce the transport-diffusive equation that will be the subject of our study, as well as the geometric tools that are needed to this purpose.
 
Let $\N$ be a $N$-dimensional smooth manifold and $\M$ a smooth $M$-dimensional manifold with a $C^2$ Riemannian metric $g$. It is assumed that $\N$ and $\M$ are globally diffeomorphic to the torus $\mathbb{T}^N$ and to $\R^M$, respectively. (See also~\cite{leb} for a formulation of Fokker-Planck equations on the cotangent bundle of general manifolds.) Let 
\[
h:\R\times\N\times\M\to [0,\infty)
\]
satisfy an evolution equation of the following form: 
\begin{equation}\label{maineq}
\partial_t h+Th=\Delta_p h+ Wh\:.
\end{equation}
Here $h=h(t,x,p)$, where $x=(x^1,\dots,x^N)$, $p=(p^1,\dots,p^M)$ are global coordinates on $\N$ and $\M$ respectively; $\partial_t$ denotes the partial derivative with respect to $t\in\R$, while $\partial_{x^I}$, $\partial_{p^i}$ denote  the partial derivatives in the coordinates $(x^I, p^i)$. Capital Latin indexes run from 1 to $N$, small Latin indexes run from $1$ to $M$.  Denoting $\X(\mathcal{B})$, $\X_*(\mathcal{B})$ the set of smooth vector fields and one form fields on a manifold $\mathcal{B}$, then $T\in\X(\N\times\M)$, the {\it transport field}, whereas $W\in\X(\M)$. Finally, $\Delta_p$ denotes the Laplace-Beltrami operator on $(\M,g)$. A subscript $p$ is attached to differential operators that act on the variables $p^1,\dots,p^M$ only. 

In order to specify the exact form of the fields $T,W$, some basic facts from differential geometry are required. In the following discussion, which is based mainly on~\cite{oneill}, we consider only (smooth, time dependent) tensor fields defined on $\M$ or $\N$, possibly obtained by projecting tensor fields from $\N\times\M$. We also remark that in the rest of the paper  we do not distinguish between a tensor field defined on $\N$ or $\M$ and its lift on $\N\times\M$.
Let $X_{(I)}$, $P_{(i)}$ denote the frame vector fields basis associated with the coordinates $x^I$, $p^i$ (i.e., $X_{(I)}f=\partial_{x^I}f$, $P_{(i)}f=\partial_{p^i}f$, for all smooth functions $f$ on $\N\times\M$) and $X_*^{(I)}, P_*^{(i)}$ their dual one form fields; clearly $P_{(i)}$ and $P_*^{(i)}$ are metrically equivalent: $g(P_{(i)},Y)=P_*^{(i)}(Y)$, for all $Y\in\X(\M)$. Note that the indexes in round brackets are list indexes and not component indexes, that is to say, for each fixed $i$, $P_{(i)}$ is a geometric object of the same type (a vector field). Now let $v^{(1)},\cdots v^{(N)}$ denote a set of $C^3$ real valued functions on $\M$. We assume that the transport vector field has the following form:
\[
T=v^{(I)}(p)X_{(I)}\:.
\]
We adopt the Einstein summation rule, whereby the sum over repeated indexes is understood.  

\begin{Remark}
When $N=M$ (or more generally when $N\leq M$)
the functions $v^{(I)}$ can be thought of as the (non-zero) components of a vector field over $\M$, the {\it velocity field}. However this interpretation is not necessary and in general not very useful, so we will refrain from adopting it. In particular, the use of the list index $(I)$ in $v^{(I)}$ reminds that this is a scalar function, which affects how geometric differential operators act on it.    
\end{Remark}

For any tensor field $R$ over $\M$, $\nabla_p  R$ denotes the covariant differential of $R$, where $\nabla_p$ is the Levi-Civita connection associated with $g$ (i.e.,$\nabla_p$ is symmetric and $\nabla_p\,g=0$). For a scalar function $f$ on $\M$, $\nabla_p  f $ is the one form $\nabla_p  f(Y)=Y(f)$, for all $Y\in\X(\M$). Any vector field $Z\in\X(\M)$ is metrically equivalent to the one form field $Z_*\in\X_*(\M)$ given by $Z_*(Y)=g(Y,Z)$, for all $Y\in\X(\M)$. The vector field metrically equivalent to $\nabla_p  f$ is the gradient of $f$, which we denote $\partial_p  f$:
\[
g(\partial_p f,Y)=\nabla_p  f (Y)\:,\ \text{or }\ \nabla_p  f=(\partial_p  f)_*\:.
\]
Using the components  $g_{ij}=g(P_{(i)}, P_{(j)})$ of the metric in the base $P_*^{(i)}\otimes P_*^{(j)}$ of the space of type (2, 0) tensor fields, we may express the action of the Laplace-Beltrami operator on scalar functions as 
\begin{equation}\label{laplacecoord}
\Delta_p f=\frac{1}{\sqrt{|g|}}\partial_{p^i}\left(\sqrt{|g|}\,g^{ij}\partial_{p^j} f\right)\:,
\end{equation}
where $g^{ij}$ is the matrix inverse of $g_{ij}$, i.e., $g^{ik}g_{kj}=\delta^i_{\ j}$ and $|g|=\det g$.
There is however a more convenient way to express $\Delta_p f$. For this we recall that the divergence of a vector field $Z\in\X(\M)$ is the contraction of $\nabla_p  Z$, i.e.,
\[
\Div Z=\nabla_p  Z (P_*^{(i)},P_{(i)})\:.
\]
We have the well known formula
\[
\Delta_p f =\Div(\partial_p  f)\:.
\]
Moreover by Stokes theorem 
\begin{equation}\label{stokes}
\int_{\R^M}\Div Z\,\sqrt{|g|}\,dp=0\:,
\end{equation}
for all vector fields $Z\in H^1 (\R^M,\sqrt{|g|}\,dp)$.

Next the definitions of the gradient of a vector field and of the divergence of a second order tensor will be recalled.  Let $\nabla_p  Z:\X_*(\M)\times\X(\M)\to\R$ be the covariant differential of $Z\in\X(\M)$. The metrically equivalent type (2, 0) tensor field $\partial_p Z: \X_*(\M)\times\X_*(\M)\to\R$ given by
\[
\partial_p Z(X_*,Y_*)=\nabla_p  Z(X_*,Y)
\] 
is called the gradient of $Z$. Moreover, given any type (2, 0) tensor field $R$, its divergence is defined as the contraction of $\nabla_p  R$ in the second and third variable, i.e.,
\[
\Div R(Y_*)=\nabla_p  R(Y_*,P_*^{(i)},P_{(i)})
\]  
and thus it is a vector field on $\M$. The following lemma collects some useful identities on the geometric objects defined above.
\begin{Lemma}\label{identities}
Let $f,f_1,f_2$ be smooth real valued functions on $\M$ and $X,Y,Z\in\X(\M)$. Then
\begin{itemize}
\item[(i)] $\partial_p (f_1f_2)=f_1\partial_p  f_2+f_1\partial_p f_2$;
\item[(ii)] $\Div(fZ)=g(\partial_p f,Z)+f\,\Div Z$;
\item[(iii)] $\Div(\partial_p f_1\otimes\partial_p f_2)(Y_*)=\partial_p ^2f_1(\nabla_p  f_2,Y_*)+\Delta_p f_2\,\partial_p  f_1(Y_*)$;
\item[(iv)] $g(X,\partial_p (g(Y,Z)))=\partial_p  Y(Z_*,X_*)+\partial_p  Z(Y_*,X_*)$;
\item[(v)] $g(Y,\Div\partial_p ^2f)=g(Y,\partial_p (\Delta_p f))+\Ric(Y,\partial_p f)$,
\end{itemize}
where $\Ric$ denotes the Ricci curvature tensor of $g$.
\end{Lemma}
\begin{proof}
The proofs of (i)--(iii) are straightforward. The identity (iv) is a consequence of Koszul's formula applied to the Levi-Civita connection; the proof can be found in~\cite[Ch.~3, Th.~11]{oneill}. The identity (v) is a direct consequence of the definition of the Riemann tensor and is proved for instance in~\cite[Lemma~1.45]{chow}.
\end{proof}
We can now define the vector field $W$.
Let $E:\R^M\to \R$, $E\in C^2$, such that
\[
\Theta^{-1}=\int_{\R^M}e^{-E}\,dp
\]
is bounded. Let $d\mu=\Theta e^{-E}dp$, a probability measure on $\R^M$. Then 
\begin{equation}\label{u}
d\mu=\Theta\,u\,\sqrt{|g|}\,dp\:,\quad\text{where }\
u=\frac{e^{-E}}{\sqrt{|g|}}\:.
\end{equation}
The triple $(\M,g,d\mu)$ is a measure metric space.
Let $L$ be the operator in the r.h.s. of~\eqref{maineq}, that is
\[
Lh=\Delta_p h+Wh\:.
\]
We require the field $W$ to be such that $L$ is symmetric in the Hilbert space $L^2(d\mu):=L^2(\R^M,d\mu)$, i.e.,
\begin{equation}\label{selfad}
\int_{\R^M}hLf\,d\mu=\int_{\R^M}fLh\,d\mu\:.
\end{equation}
\begin{Lemma}
The identity~\eqref{selfad} is verified if and only if $W=\partial_p \log u$, or equivalently, $W_*=\nabla_p \log u$.
\end{Lemma}
\begin{proof}
We have
\[
\int_{\R^M}hLf\,d\mu=\Theta\int_{\R^M}h\,(\Delta_p f)\, u\,\sqrt{|g|}\,dp+\int_{\R^M}hWf\,d\mu\:.
\]
In the previous equation we use (i)-(ii) of Lemma~\ref{identities} to get
\[
\Div(h\,u\,\partial_p f )=h\,u\,\Delta_p f+h\, g(\partial_p f,\partial_p u)+u\,g(\partial_p f,\partial_p  h)
\]
and so doing we obtain, by~\eqref{stokes},
\begin{equation}\label{temp}
\int_{\R^M}hLf\,d\mu=\int_{\R^M}h(Wf-g(\partial_p f,\partial_p \log u))\,d\mu-\int_{\R^M}g(\partial_p f,\partial_p  h)\,d\mu\:.
\end{equation}
Again we use 
\[
\Div(f\,u\,\partial_p h )=f\,u\,\Delta_p h+f\, g(\partial_p h,\partial_p u)+u\,g(\partial_p f,\partial_p  h)
\]
and so we obtain
\[
\int_{\R^M}hLf\,d\mu=\int_{\R^N}h(W-\partial_p \log u)f\,d\mu+\int_{\R^M}f (\Delta_p +\partial_p \log u)h\,d\mu\:,
\]
which implies the claim.
\end{proof}
We conclude this section by proving some integration by parts formulas.
\begin{Lemma}\label{intbyparts}
The following identities hold true, for all smooth real valued functions $f,h$ on $\M$: 
\begin{align}
&\int_{\mathbb{T}^N}hTf\,dx=-\int_{\mathbb{T}^N} fTh\,dx\:, \label{IPFT}\\
&\int_{\R^M}hLf\,d\mu=-\int_{\R^M}g(\partial_p f,\partial_p  h)\,d\mu\:.\label{IPF1}
\end{align}
\end{Lemma}
\begin{proof}
The proof of~\eqref{IPFT} is straightforward. The identity~\eqref{IPF1} follows by setting $W=\partial_p\log u$ in~\eqref{temp}. 
\end{proof}
For the next result we need to recall the definition of inner product of second order tensor fields. Given a type $(2,0)$ tensor field $R$ and a type $(0,2)$ tensor field $S$, the inner product $R\cdot S=S\cdot R$ is defined as
\[
R\cdot S=(R\otimes S)(P_*^{(i)}, P_*^{(j)},P_{(i)}, P_{(j)})\:.
\]
Componentwise this means $R\cdot S=R^{ij}S_{ij}$.
\begin{Lemma}
For all type $(2,0)$ tensor fields $A$ and $Z\in\X(\M)$ we have
\begin{equation}\label{IPF3}
\int_{\R^M}g(Z,\Div A)\,d\mu=-\int_{\R^M}A\cdot \nabla_p  Z_*\,d\mu-\int_{\R^M}A(W_*,Z_*)\,d\mu\:.
\end{equation} 
\end{Lemma}
\begin{proof}
Consider the vector field $Y\in\X(\M)$ defined by $Y(\cdot)=A(uZ_*,\cdot)$. By the Leibnitz identity,
\[
\Div Y=g(uZ,\Div A)+A\cdot\nabla_p(uZ_*)\:.
\]
Replacing in the l.h.s. of~\eqref{IPF3} we obtain
\begin{align*}
\int_{\R^M}g(Z,\Div A)\,d\mu=&\ \Theta\int_{\R^M}g(uZ,\Div A)\sqrt{|g|}\,dp\\
=&-\Theta\int_{\R^M}A\cdot\nabla_p(uZ_*)\sqrt{|g|}\,dp\\
=&-\int_{\R^M}A\cdot (W_*\otimes Z_*)\,d\mu-\int_{\R^M}A\cdot\nabla_pZ_*\,d\mu\:,
\end{align*}
which is the claim.
\end{proof}

\section{Main result}\label{mainresult}
We begin by stating our assumptions on the functions $v^{(I)}, E$ and the metric $g$. Let us recall that the Bakry-Emery-Ricci tensor is defined by
\begin{equation}\label{BEricci}
\widetilde{\Ric}=\Ric-\nabla_p  W_*=\Ric-\nabla^2_p\log u\:,
\end{equation}
where $\nabla^2_p f$ denotes the Hessian of $f$ and $u$ is the function~\eqref{u}. 
As already mentioned in the Introduction, Barky and Emery proved in~\cite{BE,B} that spatially homogeneous solutions of~\eqref{maineq} converge exponentially fast in time to the equilibrium state $h_\infty\equiv 1$ in the entropic sense (i.e., the entropy functional decays exponentially) if the tensor $\widetilde{\Ric}$ is bounded below by a constant times the metric $g$. In the spatially inhomogeneous case, we also need a bound on $\widetilde{\Ric}$ from above.  
\begin{assumption}\label{cbcass}
There exist two constants $\sigma_2\geq \sigma_1\geq 0$ such that
\begin{equation}\label{curbound}
\sigma_1 g(X,X)\leq\widetilde{\Ric}(X,X)\leq \sigma_2 g(X,X)\:,\quad\text{ for all }X\in\X(\M)\:.
\end{equation}
\end{assumption}
We denote 
\begin{equation}\label{sigma}
\sigma=\sigma_2-\sigma_1\geq 0\:.
\end{equation}
\begin{Remark}
An important example is when $\sigma=0$. In this case there exists $\lambda\in\R$ such that $\widetilde{\Ric}=\lambda g$. Solutions of the latter equation are called {\it Ricci solitons}. They play a fundamental role in the analysis of the Ricci flow heat equation, see~\cite{chow}.
\end{Remark}
Before stating the next assumption, it is convenient to give the following definition.
\begin{Definition}\label{modifiedgradient}
Given a real valued function $f$ on $\N\times\M$, and a point $x\in\N$, we denote $\mathscr{A}_xf$ the vector field over $\M$ given by
\[
\mathscr{A}_x f=(\partial_{x^I}f)\partial_p v^{(I)}\:,\quad\text{evaluated at $x\in\N$}.
\]
The metrically equivalent one form field is given by $(\mathscr{A}_x f)_*=(\partial_{x^I}f)\nabla_p  v^{(I)}$.
\end{Definition}

We emphasize that $\mathscr{A}_xf \in\X(\M)$. (More precisely, $\mathscr{A}_xf\in\X(\{x\}\times\M)\simeq\X(\M)$.) Its components in the vector fields basis $P_{(i)}$ are given by 
\[
(\mathscr{A}_xf)^i=g^{ij}{\partial_{p^j}}v^{(I)}\partial_{x^I}f\:.
\] 
For the Fokker-Planck equation~\eqref{kineticFP2}, the manifold $\M$ can be identified with the tangent space at all points $x\in\N$ and $\mathscr{A}_x f$ coincides with $\partial_x f$, the gradient in $x$ of $f$.    

Now let us define a symmetric bilinear form $A$ on $\R^N\times\R^N$ by
\[
A(\xi,\eta)=A^{IJ}\xi_I\eta_J\:,\quad A^{IJ}=g(\partial_p v^{(I)},\partial_p  v^{(J)})\:,\qquad \xi,\eta\in \R^N\:.
\]
Note that 
\begin{equation}\label{gobs}
A^{IJ}\partial_{x^I}h\partial_{x^J}h=g(\mathscr{A}_x h,\mathscr{A}_x h)\:.
\end{equation}
\begin{assumption}\label{gpos} We assume that $A$ is positive definite, 
\[
A(\xi,\xi)>0\:,\quad \text{ for all } 0\neq\xi\in\R^N\:.
\]
\end{assumption}
In the next assumption we require the validity of a (weighted) logarithmic Sobolev inequality in both variables $(x,p)$. Precisely, define the entropy functional
\begin{equation}\label{entropydef}
\mathcal{D}[h]=\int_{\mathbb{T}^N\times\R^M}h\log h\,dx\,d\mu\:,
\end{equation}
and 
\begin{subequations}\label{ixxipp}
\begin{align}
&\I_{xx}[h]=\int_{\mathbb{T}^N\times\R^M}\frac{g(\mathscr{A}_xh,\mathscr{A}_xh)}{h}\,dx\,d\mu\:,\\
&\I_{pp}[h]=\int_{\mathbb{T}^N\times\R^M}\frac{g(\partial_ph,\partial_ph)}{h}\,dx\,d\mu\:.
\end{align}
\end{subequations}
\begin{assumption}\label{logsobineqass}
We assume that for all smooth functions $h:\N\times\M\to (0,\infty)$, such that 
\[
\int_{\mathbb{T}^N\times\R^M}h\,dx\,d\mu=1\:,
\]
there exists $\alpha>0$ such that the following inequality holds:
\begin{equation}\label{logsobin}
\mathcal{D}[h]\leq \frac{1}{2\alpha}\big(\I_{xx}[h]+\I_{pp}[h]\big)\:.
\end{equation}
\end{assumption}
Logarithmic Sobolev inequalities are extensively studied in the literature and several criteria for their validity have been found. In Appendix we give a sufficient condition for the validity of~\eqref{logsobin}, which, in the spirit of our approach, takes the form of a geometric inequality on an auxiliary metric defined on the product manifold $\N\times\M$. Moreover, as already mentioned in the Introduction, the lower bound in~\eqref{curbound} for the Bakry-Emery-Ricci tensor implies that~\eqref{logsobin} holds when $h$ is independent of $x$ (i.e., when $\I_{xx}=0$).

The previous assumptions suffice if the metric $g$ and the velocity field $v$ are such that $\nabla^2_pv^{(I)}=0$. If this is not the case we need more assumptions, which we give after the following definitions.

\begin{Definition}
Given a real valued function $f$ on $\N\times\M$, and a point $x\in\N$, we denote $\mathscr{B}_xf$ the vector field over $\M$ given by
\[
\mathscr{B}_xf=(\partial_{x^I}f)\Div\partial_p^2 v^{(I)}\:,\quad\text{evaluated at $x\in\N$}.
\]
\end{Definition}
\begin{Definition}
Given a real valued function $f$ on $\N\times\M$, and a point $x\in\N$, we denote $\mathscr{C}_xf$ the type (2,0) tensor field over $\M$ given by
\[
\mathscr{C}_xf=(\partial_{x^I}f)\partial_p^2 v^{(I)}\:,\quad\text{evaluated at $x\in\N$}.
\]
The metrically equivalent type (0,2) tensor field is $(\mathscr{C}_xf)_*=(\partial_{x^I}f)\nabla_p^2  v^{(I)}$.
\end{Definition}
Next let $B, C$ denote the symmetric bilinear forms on $\R^N\times \R^N$ given by
\begin{align*}
&B(\xi,\eta)=B^{IJ}\xi_I\eta_J\:,\quad B^{IJ}=g(\Div\partial_p ^2v^{(I)}, \Div\partial_p ^2v^{(J)})\:,\\
&C(\xi,\eta)=C^{IJ}\xi_I\eta_J\:,\quad C^{IJ}=\partial_p ^2v^{(I)}\cdot\nabla_p ^2 v^{(J)}\:,
\end{align*}
and observe that
\begin{equation}\label{ABobs}
B^{IJ}\partial_{x^I}h\partial_{x^J}h=g(\mathscr{B}_xh,\mathscr{B}_xh)\:,\quad C^{IJ}\partial_{x^I}h\partial_{x^J}h=\mathscr{C}_xh\cdot(\mathscr{C}_x h)_*\:
\end{equation}
\begin{assumption}\label{ABass}
We assume that there exist two constants $\beta,\gamma\geq0$ such that
\[
B(\xi,\xi)\leq \beta A(\xi,\xi)\:,\qquad C(\xi,\xi)\leq \gamma A(\xi,\xi)\:,\quad\text{for all }\xi\in\R^N\:.
\] 
\end{assumption}
The previous assumptions suffice if $W$ lies in the kernel of the Hessian matrix of $v^{(I)}$, i.e., $\nabla^2_pv^{(I)}(W,\cdot)=0$, for all $I=1,\dots,N$. If this is not the case, we need a last assumption. Define the vectors $K^{(1)},\dots K^{(N)}$ by
\[
K^{(I)}(\cdot)=\partial^2_pv^{(I)}(W_*,\cdot)\:,
\]
i.e., componentwise, 
\[
K^{(I)}_i=(\nabla_p^2 v^{(I)})_{ij}W^j\:.
\]
Let $R$ denote the bilinear form on $\R^N\times\R^N$ given by
\[
R(\xi,\eta)=R^{IJ}\xi_I\eta_J\:,\quad R^{IJ}=g(K^{(I)},K^{(J)})\:.
\]
\begin{assumption}\label{Wass}
We assume that there exists a constant $\omega>0$ such that
\begin{equation}\label{boundW}
R(\xi,\xi)\leq \omega A(\xi,\xi)\:,\quad\text{for all }\xi\in\R^N.
\end{equation}
\end{assumption}
\begin{Remark}
For the Fokker-Planck equation~\eqref{kineticFP2} we have the following identifications: all indexes (small and capital) run from 1 to $N=M$ and
\begin{align*}
&g_{ij}=\delta_{ij}\:,\quad v^{(I)}(p)=p^I\:,\quad E=|p|^2/2\:;\\
& W^i=-p^i\:,\quad\widetilde{\Ric}_{ij}=\delta_{ij}\:,\quad \partial_p^2v^{(I)}=0\:;\\
&\mathscr{A}_xf=\partial_xf\:,\quad\mathscr{B}_xf=\mathscr{C}_xf=0\:,\quad\text{for all functions $f$}\:;\\
&A_{IJ}=\delta_{IJ}\:,\quad B=C=R=0\:.
\end{align*}
The constants in the assumptions can be chosen as $\sigma_1=\sigma_2=\alpha=1,\beta=\gamma=\omega=0$. A much less obvious example of model satisfying all the previous assumptions is given in the next section.
\end{Remark}

Before any claim on the asymptotic time behavior of solutions could be made, one has to ensure that the Cauchy problem for~\eqref{maineq} is globally well-posed. %
In  Appendix A we prove a global existence and uniqueness theorem in the $L^1$ setting, namely 
\[
0\leq h\in C([0,\infty);L^1(dx\,d\mu))\:;
\]
we restrict to the case when the dimensions of $\N$ and $\M$ are the same, i.e., $N=M$ (which is the most interesting case for the applications) and the metric components grow slower  than $|p|^2$ at infinity. The core of the proof is a generalization of the argument in~\cite[Ch. 5]{helffer}, which consists in using the hypoellipticity of the Fokker-Planck operator to prove that it is the generator of a dissipative semigroup.
We remark that it is possible to prove global well-posedness of the Cauchy problem in a much larger class (see the results in~\cite{hypo} for the classical Fokker-Planck equation~\eqref{kineticFP}), however this is beyond the purpose of this paper.   
In the following we assume that the initial datum $h_\mathrm{in}$ belongs to the space $L^1\cap L\log L(dx\,d\mu)$, that $\I_{xx}[h_\mathrm{in}]+\I_{pp}[h_\mathrm{in}]$ is bounded---see~\eqref{ixxipp}---and that $h_\mathrm{in}$ is normalized to a probability distribution:
\begin{equation}\label{normone}
\|h_\mathrm{in}\|_{L^1(dxd\mu)}:=\int_{\mathbb{T}^N\times \R^M} h_\mathrm{in}\vol=1\:.
\end{equation} 
Our main result is the following.
\begin{Theorem}\label{maintheo}
Let the Assumptions~\ref{cbcass}--\ref{Wass} be verified (Assumptions~\ref{cbcass}--\ref{ABass} suffice when $\nabla_p^2v^{(I)}(W,\cdot)\equiv 0$ and Assumptions~\ref{cbcass}--\ref{logsobineqass} suffice when $\nabla_p^2v^{(I)}\equiv 0$) and let the initial datum satisfy the aforementioned properties. There exists two constants $C>0$, $\lambda>0$, depending on the parameters $\sigma_1,\sigma_2,\omega,\alpha,\beta,\gamma$, and which can be explicitly computed, such that the entropy functional~\eqref{entropydef}
satisfies 
\begin{equation}\label{entropydecay}
\mathcal{D}[h](t)\leq C(\I_{xx}[h_\mathrm{in}]+\I_{pp}[h_\mathrm{in}]) e^{-\lambda t}\:.
\end{equation}
\end{Theorem}  

\begin{Remark}
Using~\eqref{CKineq} we have 
\[
\|h-1\|_{L^1(dxd\mu)}=O(e^{-\lambda\,t/2})\:,\quad\text{ as $t\to\infty$}\:.
\]
Equivalently, the solution of~\eqref{kineticFP2} converges to the steady state $f_\infty\sim e^{-E}$ with exponential rate in the $L^1$ norm.
\end{Remark}

\begin{Remark}
Of course there is no loss of generality in restricting to initial data that satisfy~\eqref{normone}, since~\eqref{maineq} preserves the $L^1(dxd\mu)$-norm and is invariant by the rescaling $h\to Mh$. Solutions with mass $M>0$ converge to the equilibrium state $h_\infty=M$.
\end{Remark}

\begin{Remark}
It should be emphasized that the entropy method is not suitable to obtain the optimal constants $C,\lambda$ for which~\eqref{entropydecay} holds and in some cases may lead to an exponential rate very far from the real one, see~\cite{vil}. The hypoelliptic techniques used in~\cite{Herau} are better for this purpose. 
\end{Remark}
The proof of the Theorem~\ref{maintheo} is to be found in Section~\ref{proof}. In the next section we show that the result applies to an important physical model: the relativistic Fokker-Planck equation.

\section{Application to the relativistic Fokker-Planck equation}\label{appl}
The relativistic kinetic Fokker-Planck equation is obtained from~\eqref{kineticFP2} by setting 
\begin{equation}\label{relFP}
v(p)=\frac{p}{\sqrt{1+|p|^2}}\:,\quad D=\frac{\mathbb{I}+p\otimes p}{\sqrt{1+|p|^2}}\:,\quad E=\theta\sqrt{1+|p|^2}\:.
\end{equation}
We restrict to the three dimensional problem: $x\in\mathbb{T}^3$, $p\in\R^3$ (thus capital and small Latin indexes run both from 1 to 3 in this section).
$v(p)$ is the relativistic velocity, $D$ is the relativistic diffusion matrix and $\sqrt{1+|p|^2}$ is the relativistic energy. We set the rest mass of the particles and the speed of light equal to one. The equilibrium state is given by the J\"uttner distribution
\begin{equation}\label{juttner}
\mathcal{J}_\theta(p)=Ze^{-\theta\sqrt{1+|p|^2}}\:,
\end{equation}
where $Z$ is a constant (fixed by the mass of the system) and $\theta$ is a positive parameter which, up to a dimensional constant, coincides with $1/T$, where $T$ is the temperature of the surrounding bath in which the particles are moving. Although the interest in the relativistic Fokker-Planck equation has increased substantialy in recent years~\cite{AC0,DH2,H1}, the relativistic theory of Brownian motions is an old classical topic~\cite{Du}. In this section we prove that solutions of the relativistic Fokker-Planck equation converge with exponential rate to the J\"uttner equilibrium, provided the parameter $\theta$ is sufficiently large, i.e., for a sufficiently small  temperature. We do so by showing that all the assumptions of Theorem~\ref{maintheo} are satisfied. To this purpose we first normalize the solution by introducing $h=fe^E$ and rewrite the relativistic Fokker-Planck equation in the form~\eqref{maineq}, where the metric $g$ is given by
\[
g_{ij}=p_0(\delta_{ij}-\frac{p_ip_j}{p_0^2})\:,\qquad p_0=\sqrt{1+|p|^2}\:.
\]
Note that $g$ is conformal to the hyperbolic metric $\delta_{ij}-p_ip_j/p_0^2$ and that the matrix inverse of $g_{ij}$ is 
\[
g^{ij}=\frac{1}{p_0}(\delta^{ij}+p^ip^j)\:,
\]
where the indexes of the variables $p_i$ are raised and lowered with the Euclidean matrix, e.g., $p_i=\delta_{ij}p^j$.
Since $\det g= p_0$, the function $u=e^{-E}/\sqrt{\det g}$ is given by
\[
u=\frac{e^{-\theta p_0}}{\sqrt{p_0}}\:.
\]
We begin by showing the validity of Assumptions~\ref{gpos} and~\ref{ABass}, which are independent of the energy  function $E$. A straightforward calculation shows that the bilinear form $A^{IJ}$ is given by
\[
A^{IJ}=g(\partial_pv^{(I)},\partial_pv^{(J)})=g^{ij}\partial_{p^i}v^{(I)}\partial_{p^j}v^{(J)}=\frac{1}{p_0^3}(\delta^{IJ}-\frac{p^Ip^J}{p_0^2})\:.
\]
Thus Assumption~\ref{gpos} holds, because
\begin{equation}\label{estAw}
A(\xi,\xi)=\frac{1}{p_0^3}(|\xi|^2-\frac{(p\cdot\xi)^2}{p_0^2})\geq \frac{1}{p_0^3}|\xi|^2(1-\frac{|p|^2}{p_0^2})=\frac{|\xi|^2}{p_0^5}>0\:,
\end{equation}
for all $0\neq\xi\in\R^3$.
The bilinear form $B^{IJ}$ and $C^{IJ}$ are given by\footnote{The remaining calculations in this section have been carried out with MATHEMATICA.}
\begin{align*}
&B^{IJ}=\frac{496 p_0^6 - 9030 p_0^4 + 1035 p_0^2 - 25}{16 p_0^{13}} \delta^{IJ}\\
&\qquad\ \,+ \frac{25 - 1035 p_0^2 + 10551 p_0^4 + 1610 p_0^6 + 729 p_0^8}{16 
     p_0^{10}}A^{IJ}\:,\\
&C^{IJ}=\frac{9}{4p_0^6}\delta^{IJ}+\frac{9(2p_0^2-3)}{4p_0^3}A^{IJ}\:.
\end{align*}
Since, by~\eqref{estAw}, 
\begin{equation}\label{estA}
|\xi|^2\leq p_0^5A(\xi,\xi)
\end{equation} 
holds, we have
\[
B(\xi,\xi)\leq \frac{P_8(p_0)}{p_0^{10}}A(\xi,\xi)\:,
\]
where $P_8(p_0)$ is a polynomial of degree 8. Thus it is clear that there exists $\beta>0$ such that $B(\xi,\xi)\leq \beta A(\xi,\xi)$ and by the same argument, there exists $\gamma>0$ such that  $C(\xi,\xi)\leq\gamma A(\xi,\xi)$.
We conclude that Assumption~\ref{ABass} is satisfied as well. As to Assumption~\ref{Wass}, the bilinear form $R^{IJ}$ is given by
\[
R^{IJ}=\frac{(1 + 2 \theta p_0)^2}{16 p_0^9} (16 (p_0^2 - 1)\delta^{IJ} + 
   p_0^3 (9 p_0^4 - 34 p_0^2 + 25)A^{IJ})
\]
and arguing as before it is easy to prove that $R(\xi,\xi)\leq\omega A(\xi,\xi)$, for some positive $\omega$.
Let us now take care of Assumption~\ref{cbcass}. The Ricci tensor of $g$ and the Hessian of $\log u$ are given by
\[
\mathrm{Ric}_{ij}=\frac{1}{4 p_0^2} (3\delta_{ij} - \frac{4 + 15 p_0^2}{p_0} g_{ij})\:,
\]
\[
(\nabla^2_p\log u)_{ij}=\frac{1}{4 p_0^2} ((4 + 4 \theta p_0)\delta_{ij} - \frac{3 + 3 p_0^2 + 2 \theta p_0(1 + 3 p_0^2)}{p_0}
  g_{ij})\:.
\]
Therefore the Bakry-Emery-Ricci tensor reads
\[
\widetilde{\mathrm{Ric}}_{ij}=\frac{1}{4 p_0^2} (-(1 + 4 \theta p_0)\delta_{ij} + 
   \frac{6 \theta p_0^3 - 12 p_0^2 + 2 \theta p_0 - 1}{p_0}g_{ij})\:.
\]
It is straightforward that the bound from above in~\eqref{curbound} is satisfied, for all $\theta>0$. However the lower bound is satisfied if and only if $\theta$ is sufficiently large. To see this we use that 
\begin{equation}\label{estg}
|X|^2\leq p_0g(X,X)\:,\quad\text{ for all $X\in\R^3$}\:, 
\end{equation}
whence
\[
\widetilde{\mathrm{Ric}}(X,X)\geq \left[\frac{2\theta p_0^3-13p_0^2+2\theta p_0-1}{4p_0^3}\right] g(X,X)
\]
and the minimum of the function on square brackets is strictly positive if and only if $\theta$ is sufficiently large (e.g., $\theta\geq 4$ suffices).
It remains to check the validity of the logarithmic Sobolev inequality~\eqref{logsobin}, where in this case
\begin{align*}
&\I_{xx}[h]=\int_{\mathbb{T}^3\times\R^3}\frac{1}{p_0^3h}\left(|\nabla_x h|^2-\frac{|p\cdot\nabla_xh|^2}{p_0^2}\right)dx\,d\mu\:,\\
&\I_{pp}[h]=\int_{\mathbb{T}^3\times\R^3}\frac{1}{p_0h}\left(|\nabla_ph|^2+|p\cdot\nabla_ph|^2\right)dx\,d\mu\:.
\end{align*}
To prove~\eqref{logsobin} we use Theorem~\ref{sufflogsob} in Appendix B. We introduce the metric
\[
G=G_{IJ}dx^Idx^J+g_{ij}dp^idp^j\:,
\]
where $G_{IJ}=A_{IJ}=p_0^3(\delta_{IJ}+p_Ip_J)$ is the inverse matrix of $A^{IJ}$, 
the function
\[
U=\frac{u}{\sqrt{\det(A_{IJ})}}=\frac{u}{p_0^{11/2}}
\]
and check that there exists a constant $\alpha>0$ such that
\begin{equation}\label{curboundcondG}
\mathrm{Ric}^G(Z,Z)-\nabla_p^2\log U(Z,Z)\geq\alpha G(Z,Z)\:.
\end{equation}
The Ricci tensor of $G$ is given by
\begin{align*}
&\mathrm{Ric}^G_{IJ}=\frac{13}{2} p_0^2\delta_{IJ} - \frac{19 p_0^2 - 7}{p_0^3}G_{IJ}\:,\\
&\mathrm{Ric}^G_{iJ}=0\:,\\
&\mathrm{Ric}^G_{ij}=\frac{3}{2 p_0^2}\delta_{ij} - \frac{25 p_0^2 - 3}{2 p_0^3}g_{ij}\:.
\end{align*}
The Hessian of $\log U$ is
\begin{align*}
&(\nabla^2\log U)_{IJ}=\frac{23+2\theta p_0}{4}(2p_0^2\delta_{IJ}-\frac{5p_0^2-2}{p_0^3}G_{IJ})\:,\\
&(\nabla^2\log U)_{iJ}=0\:,\\
&(\nabla^2\log U)_{ij}=\frac{23+\theta p_0}{p_0^2}\delta_{ij}-\frac{6\theta p_0^3+69p_0^2+2\theta p_0+69}{4p_0^3}g_{ij}.
\end{align*}
Using the previous formulas, the bound~\eqref{estg} and
\begin{equation}\label{estA2}
|X|^2\leq \frac{G_{IJ}X^IX^J}{p_0^3}\:,\quad\text{for all $X\in\R^3$},
\end{equation}
it is straightforward to prove that the curvature bound~\eqref{curboundcondG},  and thus the logarithmic Sobolev inequality~\eqref{logsobin}, holds when the constant $\theta$ is sufficiently large. In conclusion all the assumptions of Theorem~\ref{maintheo} are satisfied for the relativistic Fokker-Planck equation provided the constant $\theta$ is large enough, or, equivalently, the temperature of the surrounding bath is sufficiently small. Therefore Theorem~\ref{maintheo} yields the following result.
\begin{Theorem}
Let $0\leq f_\mathrm{in}$ be an initial datum of mass $M>0$ for the the relativistic Fokker-Planck equation, i.e., for~\eqref{kineticFP2} with~\eqref{relFP} substituted in. Denote by $\mathcal{J}_{\theta,M}$ the J\"uttner distribution~\eqref{juttner} with mass $M$. Then there exists $\theta_0>0$ such that for all $\theta\geq\theta_0$ there exists two positive constants $C,\lambda$, depending on $\theta$, such that the solution $f$ of the relativistic Fokker-Planck equation satisfies
\[
\|f-\mathcal{J}_{\theta,M}\|_{L^1}\leq C e^{-\lambda t}.
\]
\end{Theorem}
\begin{Remark}
In~\cite{AC1} the precise values of the constants $C,\lambda,\theta_0$ are found in the case of spatially homogeneous (SH) solutions. Moreover it is shown that, at least within the class of SH solutions, exponential convergence to equilibrium holds for all $\theta>0$ in a suitable weighted $L^2$ norm. 
\end{Remark}
\section{Proof of the main result}\label{proof}

In the rest of the paper the following abbreviations will be used:
\[
\int \cdots\, dx\,d\mu=\int_{\,\mathbb{T}^N\times\R^M}\cdots\, dx\,d\mu\:
\]
and 
\[
\h=\log h\:.
\]
Moreover the measure $dx\,d\mu$ will be omitted in the proofs.

In the following we assume that $h$ is a positive smooth solution of~\eqref{maineq}. The proof of the  result for $L^1$ non-negative solutions with finite entropy is obtained by a standard approximation argument, see~\cite{bdol, CLR} for examples of this procedure. 

Recall that
\begin{align*}
&\mathcal{D}[h]=\int h\,\h\,dx\,d\mu\:,\\
&\I_{pp}[h]=\int g(\partial_p h,\partial_p \h)\,dx\,d\mu\:,\\
&\I_{xx}[h]=\int  g(\mathscr{A}_x h,\mathscr{A}_x \h)\,dx\,d\mu
\end{align*}
and define the mixed derivatives term
\[
\I_{xp}[h]=\int g(\mathscr{A}_x h,\partial_p\h)\,dx\,d\mu\:.
\]
Given four constants $a,b,c, k>0$, we define the modified entropy as
\[
\mathcal{E}[h]=k\,\D[h]+a\,\I_{pp}[h]+2b\,\I_{xp}[h]+c\,\I_{xx}[h]\:.
\]
We divide the proof in three subsections.
\subsection{Evolution of the modified Entropy}\label{derivative}
Our first goal is to study the time evolution of the modified entropy, by computing the time derivative of $\D$, $\I_{xx}, \I_{xp}$ and $\I_{pp}$.

\begin{Lemma}\label{dtentropy}
The following holds:
\[
\frac{d}{dt}\D[h]=-\I_{pp}[h]\:,
\]
\end{Lemma}
\begin{proof}
We compute
\[
\frac{d}{dt}\D[h]=\int \partial_t h(1+\h)=-\int (1+\h)Th+\int Lh\,dx\,d\mu+\int \h L h\:.
\]
By~\eqref{IPFT} and~\eqref{IPF1}, the first two terms vanish and
\[
\int \h L h=-\int g(\partial_p h,\partial_p \h)\:.
\]
\end{proof}

\begin{Lemma}\label{dipp}
The following holds:
\begin{align*}
\frac{d}{dt}\I_{pp}[h]=&-2\I_{xp}[h]-2\int h\,\widetilde{\Ric}(\partial_p\h,\partial_p\h)\,dx\,d\mu\\
&-2\int h\, \partial_p^2\h\cdot\nabla_p ^2\h\,dx\,d\mu\:.
\end{align*}
where $\widetilde{\Ric}$ is the Bakry-Emery-Ricci tensor~\eqref{BEricci}.
\end{Lemma}
\begin{proof}
We compute 
\begin{align*}
\frac{d}{dt}\I_{pp}[h]=&\,2\int g(\partial_p \h,\partial_p \partial_t h)-\int g(\partial_p\h,\partial_p\h)\partial_th\\
=&\underbrace{-2\int g(\partial_p\h,\partial_p (Th))}_\heartsuit
\underbrace{+\int g(\partial_p\h,\partial_p\h)Th}_\diamondsuit\\
&\underbrace{+2\int g(\partial_p\h,\partial_p (Lh))}_\clubsuit
\underbrace{-\int g(\partial_p\h,\partial_p\h)Lh}_\spadesuit\:.
\end{align*}
We claim that $\heartsuit + \diamondsuit = -2\I_{xp}[h]$. We prove this using the coordinates representation. From one hand
\[
g(\partial_p\h,\partial_p(Th))=g^{ij}\partial_{p^i}\h(\partial_{p^j}v^{(I)})\partial_{x^I}h+g^{ij}\partial_{p^i}\h\, v^{(I)}\partial_{p^j}\partial_{x^I}h\:;
\]
on the other hand, integrating by parts in the $x$ variable,
\[
\diamondsuit=-2\int g^{ij}\partial_{p^j}h\partial_{p^i}\partial_{x^I}\h\,v^{(I)}=2\int g^{ij}\partial_{p^j}\partial_{x^I}h\, v^{(I)}\partial_{p^i}\h\:.
\]
Thus
\begin{equation}\label{temporal1}
\heartsuit+\diamondsuit=-2\int g^{ij}\partial_{p^i}\h(\partial_{p^j}v^{(I)})\partial_{x^I}h=-2\int g(\mathscr{A}_x h,\partial_p\h)\:.
\end{equation}
The term $\clubsuit$ is
\[
\clubsuit=2\int g(\partial_p\h,\partial_p(\Delta_ph))+2\int g(\partial_p\h,\partial_p(Wh))=\clubsuit_1+\clubsuit_2\:.
\]
By (v) of Lemma~\ref{identities} and~\eqref{IPF3} we have
\begin{align*}
\clubsuit_1=&-2\int\Ric(\partial_p\h,\partial_p h)+2\int g(\partial_p\h,\Div\partial_p^2h)\\
=&-2\int\Ric(\partial_p\h,\partial_p h)-2\int\partial_p^2h\cdot\nabla_p^2\h-2\int \partial_p^2 h(W_*,\nabla_p\h)\:.
\end{align*}
Moreover by (iv) of Lemma~\ref{identities},
\[
\clubsuit_2=2\int g(\partial_p\h,\partial_p(g(\partial_ph,W)))=2\int\partial_p^2h(W_*,\nabla_p\h)
+2\int \partial_pW(\nabla_ph,\nabla_p\h)\:.
\]
Summing up and using the identity
\[
\partial_p^2 h=h\,\partial_p^2\h+\partial_p\h\otimes\partial_p h
\]
we obtain
\begin{equation}
\clubsuit=-2\int h\,\widetilde{\Ric}(\partial_p\h,\partial_p\h)-2\int h\,\partial_p^2\h\cdot\nabla_p^2\h-2\int\nabla_p^2\h(\partial_p\h,\partial_ph)\:.\label{temporal2}
\end{equation}
Finally, by~\eqref{IPF1} and (iv) of Lemma~\ref{identities},
\begin{equation}\label{temporal3}
\spadesuit=\int g(\partial_ph,\partial_p(g(\partial_p\h,\partial_p\h)))=2\int\partial_p^2\h(\nabla_p\h,\nabla_ph)\:,
\end{equation}
which cancels out with the last term of~\eqref{temporal2}. The claim follows summing up~\eqref{temporal1}--\eqref{temporal3}. 
\end{proof}

\begin{Lemma}\label{dixp}
The following holds:
\begin{align*}
\frac{d}{dt}\I_{xp}[h]=&-\I_{xx}[h]-\int h\,\widetilde{\Ric}(\mathscr{A}_x \h,\partial_p\h)\,dx\,d\mu\\
&-2\int h\, \partial_p^2\h\cdot\nabla_p (\mathscr{A}_x \h)_*\,dx\,d\mu
+\int g(\partial_p\h,\mathscr{B}_x h)\,dx\,d\mu\\
&+2\int h\,\partial_p^2\h\cdot(\mathscr{C}_x\h)_*\,dx\,d\mu
+\int\mathscr{C}_xh(W_*,\nabla_p\h)\,dx\,d\mu\:.
\end{align*}
\end{Lemma}
\begin{proof}
We have
\begin{align*}
\frac{d}{dt}\I_{xp}[h]=&\int g(\mathscr{A}_x\partial_th,\partial_p\h)-\int  g(\mathscr{A}_x\h,\partial_p\h)\partial_th+ \int g(\mathscr{A}_x\h,\partial_p\partial_th))\\
=&\underbrace{-\int g(\mathscr{A}_x\h,\partial_p(Th))}_{\heartsuit_1}
\underbrace{-\int g(\mathscr{A}_x(Th),\partial_p\h)}_{\heartsuit_2}\underbrace{+\int g(\mathscr{A}_x\h,\partial_p\h) Th}_{\heartsuit_3}\\
&\underbrace{-\int g(\mathscr{A}_x\h,\partial_p\h)Lh}_\diamondsuit
\underbrace{+\int g(\mathscr{A}_x\h,\partial_p(Lh))}_\clubsuit
\underbrace{+\int g(\mathscr{A}_x(Lh),\partial_p\h)}_\spadesuit\:.
\end{align*}
Now we claim that 
\begin{equation}\label{temp0}
\heartsuit=\heartsuit_1+\heartsuit_2+\heartsuit_3=-\I_{xx}[h]\:. 
\end{equation}
In fact, using the coordinates representation the first term of $\heartsuit$ can be rewritten as
\begin{align*}
\heartsuit_1&=-\int g^{ij}\partial_{p^i}v^{(I)}\partial_{x^I}\h\,\partial_{p^j}v^{(J)}\partial_{x^J}h-\int g^{ij}\partial_{p^i}v^{(I)}
\partial_{x^I}\h\, v^{(J)}\partial_{p^j}\partial_{x^J}h\\
&=\heartsuit_{1A}+\heartsuit_{1B}\:.
\end{align*}
It is clear that $\heartsuit_{1A}=-\I_{xx}$. Integrating by parts in the $x$ variable we obtain
\[
\heartsuit_{1B}= \int g^{ij}\partial_{p^i}v^{(I)}
\partial_{x^I}\partial_{x^J}\h\, v^{(J)}\partial_{p^j}h\:.
\]
In the previous expression we use the identity
\[
\partial_{x^I}\partial_{x^J}\h=h^{-1}\partial_{x^I}\partial_{x^J} h-h^{-2}\partial_{x^I}h\partial_{x^J}h
\]
and so doing we obtain
\begin{align*}
\heartsuit_{1B}=&\int g^{ij}\partial_{p^i}v^{(I)}\partial_{x^I}\partial_{x^J}h\,v^{(J)}\partial_{p^j}\h-\int g^{ij}\partial_{p^i}v^{(I)}\partial_{x^I}\h\, v^{(J)}\partial_{x^J}h\partial_{p^j}\h\\
=&\int g(\mathscr{A}_x(Th),\partial_p\h)-\int g(\mathscr{A}_x\h,\partial_p\h)Th=-\heartsuit_2-\heartsuit_3\:.
\end{align*}
This proves~\eqref{temp0}. It remains to study the integrals $\diamondsuit,\clubsuit,\spadesuit$. We begin by applying~\eqref{IPF1} and (iv) of Lemma~\ref{identities} to $\diamondsuit$:
\begin{equation}\label{temp2}
\diamondsuit= \int g(\partial_p h,\partial_p (g(\mathscr{A}_x\h,\partial_p\h)))=\!\!\int \partial_p(\mathscr{A}_x\h)(\nabla_p\h,\nabla_ph)+\!\!\int\partial_p^2\h((\mathscr{A}_x\h)_*,\nabla_ph)\:.
\end{equation}
As to $\clubsuit$, we first split it as
\[
\clubsuit=\int g(\mathscr{A}_x\h,\partial_p(\Delta_ph))+\int g(\mathscr{A}_x\h,\partial_p(Wh))=\clubsuit_1+\clubsuit_2\:.
\]
By (v) of Lemma~\ref{identities} and~\eqref{IPF3} we have
\begin{align*}
\clubsuit_1=&\int g(\mathscr{A}_x\h,\Div\partial_p^2h)-\int\Ric(\mathscr{A}_x\h,\partial_ph)\\
=&-\int\partial_p^2h\cdot\nabla_p(\mathscr{A}_x\h)_*-\int\partial_p^2h(W_*,(\mathscr{A}_x\h)_*)-\int\Ric(\mathscr{A}_x\h,\partial_ph)\:.
\end{align*}
Likewise
\[
\clubsuit_2=\int g(\mathscr{A}_x\h,\partial_p(g(\partial_ph,W)))=\int\partial_p^2h(W_*,(\mathscr{A}_x\h)_*)+\int\nabla_pW_*(\partial_ph,\mathscr{A}_x\h)\:.
\]
Summing up,
\[
\clubsuit=-\int\widetilde{\Ric}(\mathscr{A}_x\h,\partial_p h)-\int\partial_p^2h\cdot\nabla_p(\mathscr{A}_x\h)_*\:.
\]
In the second integral we replace
\[
\partial_p^2 h=h\,\partial_p^2\h+\partial_p\h\otimes\partial_p h
\]
and we get
\begin{equation}
\clubsuit=-\int h\,\widetilde{\Ric}(\mathscr{A}_x\h,\partial_p \h)-\int h\,\partial_p^2\h\cdot\nabla_p(\mathscr{A}_x\h)_*-\int\partial_p(\mathscr{A}_x\h)(\nabla_p\h,\nabla_ph)\:.\label{temp3}
\end{equation}
Note the the last term in the r.h.s. of~\eqref{temp3} cancels out with the first term in the r.h.s. of~\eqref{temp2}. We now work out the term $\spadesuit$. First we rewrite it as
\[
\spadesuit=\int L(\partial_{x^I}h)g(\partial_pv^{(I)},\partial_p\h)\:.
\]
Then by~\eqref{IPF1} and (iv) of Lemma~\ref{identities} we have
\begin{subequations}\label{tempo}
\begin{equation}\label{temp4}
\spadesuit=-\int\partial_p^2 v^{(I)}(\nabla_p\h,\nabla_p\partial_{x^I}h)-\int\partial_p^2\h(\nabla_pv^{(I)},\nabla_p\partial_{x^I} h)=\spadesuit_1+\spadesuit_2\:.
\end{equation}
In $\spadesuit_1$ we apply the identity
\[
\nabla_p\h\otimes\nabla_p\partial_{x^I}h=\nabla_p(\nabla_p\h \partial_{x^I}h)-\nabla_p^2\h\partial_{x^I}h
\]
and~\eqref{IPF3} to obtain
\begin{align}\label{temp5}
\spadesuit_1&=-\int\partial_p^2 v^{(I)}(\nabla_p\h,\nabla_p\partial_{x^I}h)=-\int\partial_p^2v^{(I)}\cdot\nabla_p\h\otimes\nabla_p\partial_{x^I}h
\nonumber\\
&=\int\partial_p^2v^{(I)}\partial_{x^I}h\cdot\nabla_p^2\h-\int\partial_p^2 v^{(I)}\cdot\nabla_p(\nabla_p\h\partial_{x^I}h)\nonumber\\
&=\int h\,\partial_p^2\h\cdot(\mathscr{C}_x\h)_*+\int g(\partial_p\h,\mathscr{B}_x h)+\int\partial_p^2 v^{(I)}(W_*,\partial_{x^I}h\nabla_p\h)\:.
\end{align}
Note that the last term in~\eqref{temp5} can be rewritten as
\begin{equation}\label{tempo5}
\int\partial_p^2 v^{(I)}(W_*,\partial_{x^I}h\nabla_p\h)=\int\mathscr{C}_xh(W_*,\nabla_p\h)\:.
\end{equation}
In $\spadesuit_2$ we apply the identity
\[
\nabla_p\partial_{x^I}h=h\nabla_p\partial_{x^I}\h+\partial_{x^I}\h\nabla_ph
\]
to obtain
\begin{align}\label{temp6}
\spadesuit_2&=-\int\partial_p^2\h(\nabla_pv^{(I)},\partial_{x^I}\h\nabla_ph)-\int h\,\partial_p^2\h(\nabla_p v^{(I)},\nabla_p\partial_{x^I}\h)\nonumber\\
&=-\int\partial_p^2\h((\mathscr{A}_x\h)_*,\nabla_ph)-\int h\,\partial_p^2\h\cdot\nabla_p v^{(I)}\otimes\nabla_p\partial_{x^I}\h=\spadesuit_{2A}+\spadesuit_{2B}\:.
\end{align}
Note that $\spadesuit_{2A}$ cancels out with the second term in the r.h.s. of~\eqref{temp2}. 
In $\spadesuit_{2B}$ we use
\[
\nabla_p v^{(I)}\otimes\nabla_p\partial_{x^I}\h=\nabla_p(\mathscr{A}_x\h)_*-\nabla_p^2v^{(I)}\partial_{x^I}\h
\]
to finally obtain
\begin{equation}\label{temp7}
\spadesuit_{2B}=-\int h\,\partial_p^2\h\cdot\nabla_p(\mathscr{A}_x\h)_*+\int h\,\partial_p^2\h\cdot(\mathscr{C}_x\h)_*\:.
\end{equation}
\end{subequations}
The claim follows by~\eqref{temp0}--\eqref{tempo}.
\end{proof}

\begin{Lemma}\label{dixx}
The following holds:
\begin{align*}
\frac{d}{dt}\I_{xx}[h]=&-2\int h\,\partial_p(\mathscr{A}_x \h)\cdot\nabla_p(\mathscr{A}_x \h)_*\,dx\,d\mu
+2\int g(\mathscr{A}_x \h,\mathscr{B}_x h)\,dx\,d\mu\\
&+4\int h\,\partial_p(\mathscr{A}_x \h)\cdot(\mathscr{C}_x h)_*\,dx\,d\mu+2\int\mathscr{C}_xh(W_*,(\mathscr{A}_x\h)_*)\,dx\,d\mu\:.
\end{align*}
\begin{proof}
The proof is very similar to that of Lemma~\ref{dixp}. First we compute
\begin{align*}
\frac{d}{dt}\I_{xx}[h]=&2\int g(\mathscr{A}_x\partial_th,\mathscr{A}_x\h)-\int g(\mathscr{A}_x\h,\mathscr{A}_x\h)\partial_th\\
=&\underbrace{-2\int g(\mathscr{A}_x(Th),\mathscr{A}_x\h)}_\heartsuit
\underbrace{+\int Th\,g(\mathscr{A}_x\h,\mathscr{A}_x\h)}_\diamondsuit\\
&\underbrace{-\int Lh\,g(\mathscr{A}_x\h,\mathscr{A}_x\h)}_\clubsuit
\underbrace{+2\int g(\mathscr{A}_x(Lh),\mathscr{A}_x\h)}_\spadesuit
\end{align*}
We claim that 
\begin{equation}\label{tempo1}
\heartsuit+\diamondsuit=0\:.
\end{equation}
In fact, by~\eqref{IPFT}
\begin{align*}
\diamondsuit&=-\int h T(g(\mathscr{A}_x\h,\mathscr{A}_x\h))=-2\int h g(\mathscr{A}_x (T\h),\mathscr{A}_x\h)\\
&=-2\int g(\mathscr{A}_x (Th),\mathscr{A}_x\h)+2\int (Th)\,g(\mathscr{A}_x\h,\mathscr{A}_x\h)\Rightarrow \diamondsuit=-\heartsuit\:.
\end{align*}
By~\eqref{IPF1} and (iv) of Lemma~\ref{identities} the term $\clubsuit$ can be rewritten as
\begin{equation}\label{tempo2}
\clubsuit=2\int\partial_p(\mathscr{A}_x\h)((\mathscr{A}_x\h)_*,\nabla_p h)\:.
\end{equation}
Likewise,
\begin{subequations}\label{tempor}
\begin{align}
\spadesuit&=2\int L(\partial_{x^I}h)\,g(\partial_p v^{(I)},\mathscr{A}_x\h)=-2\int g(\partial_ph,\partial_p(g(\partial_p v^{(I)},\mathscr{A}_x(\partial_{x^I}\h))))\nonumber\\
&=-2\int\partial_p^2v^{(I)}((\mathscr{A}_x\h)_*,\nabla_p\partial_{x^I} h)-2\int\partial_p(\mathscr{A}_x\h)(\nabla_p v^{(I)},\nabla_p\partial_{x^I} h)\nonumber\\
&=\spadesuit_1+\spadesuit_2\:.
\end{align}
Using the identity
\[
(\mathscr{A}_x\h)_*\otimes\nabla_p\partial_{x^I}h=\nabla_p((\mathscr{A}_x\h)_*\partial_{x^I}h)-\partial_{x^I}h\nabla_p(\mathscr{A}_x\h)_*\:,
\]
we may rewrite $\spadesuit_1$ as
\begin{align*}
\spadesuit_1&=-2\int\partial_p^2 v^{(I)}((\mathscr{A}_x\h)_*,\nabla_p\partial_{x^I}h)=-2\int\partial_p^2 v^{(I)}\cdot(\mathscr{A}_x\h)_*\otimes\nabla_p\partial_{x^I}h\\
&=-2\int\partial_p^2 v^{(I)}\cdot\nabla_p((\mathscr{A}_x\h)_*\partial_{x^I}h)+2\int\partial_p^2v^{(I)}\cdot\nabla_p(\mathscr{A}_x\h)_*\partial_{x^I}h\:.
\end{align*}
Applying~\eqref{IPF3} to the first term in the last line we get
\begin{equation}
\spadesuit_1=2\int g(\mathscr{A}_x\h,\mathscr{B}_x h)+2\int\partial_p^2 v^{(I)}(W_*,(\mathscr{A}_x\h)_*\partial_{x^I}h)+2\int h\,\partial_p(\mathscr{A}_x\h)\cdot(\mathscr{C}_x\h)_*\:.
\end{equation}
The second integral in the right hand side can be rewritten as
\begin{equation}
2\int\partial_p^2 v^{(I)}(W_*,(\mathscr{A}_x\h)_*\partial_{x^I}h)=2\int\mathscr{C}_x(W_*,(\mathscr{A}\h)_*)\:.
\end{equation}
Using the identity
\[
\nabla_p\partial_{x^I}h=h\nabla_p\partial_{x^I}\h+\partial_{x^I}\h\nabla_p h
\]
the term $\spadesuit_2$ becomes
\begin{align}\label{tempo4}
\spadesuit_2&=-2\int h\,\partial_p(\mathscr{A}_x\h)(\nabla_pv^{(I)},\nabla_p\partial_{x^I}\h)-2\int\partial_p(\mathscr{A}_x\h)(\nabla_p v^{(I)},\partial_{x^I}\h\nabla_ph)\nonumber\\
&=\spadesuit_{2A}+\spadesuit_{2B}\:.
\end{align}
Note that $\spadesuit_{2B}$ cancels out with $\clubsuit$. In $\spadesuit_{2A}$ we use
\[
\partial_p(\mathscr{A}_x\h)(\nabla_p v^{(I)},\nabla_p\partial_{x^I}\h)=\partial_p(\mathscr{A}_x\h)\cdot\nabla_pv^{(I)}\otimes\nabla_p\partial_{x^I}\h
\]
and
\[
\nabla_pv^{(I)}\otimes\nabla_p\partial_{x^I}\h=\nabla_p(\mathscr{A}_x\h)_*-\nabla_p^2v^{(I)}\partial_{x^I}\h
\]
to obtain
\begin{equation}
\spadesuit_{2A}= -2\int h\,\partial_p(\mathscr{A}_x\h)\cdot\nabla_p(\mathscr{A}_x\h_*)+2\int h\,\partial_p(\mathscr{A}_x\h)\cdot(\mathscr{C}_x h)_*\:.
\end{equation}
\end{subequations}
Summing up~\eqref{tempo1}--\eqref{tempor} concludes the proof.
\end{proof}
\end{Lemma}

\subsection{A differential inequality for the modified entropy}
Recall that
\[
\mathcal{E}[h]=k\,\D[h]+a\,\I_{pp}[h]+2b\,\I_{xp}[h]+c\,\I_{xx}[h]\:.
\]
In this section we prove that, under suitable conditions on the constants $a,b,c,k$, the modified entropy satisfies
\begin{equation}\label{goal}
\mathcal{E}[h]\geq k\,\D[h]\:,\qquad \frac{d}{dt}{\mathcal{E}}[h]\leq -d(a\, \I_{pp}[h]+2b\,\I_{xp}[h]+c\,\I_{xx}[h])\:,
\end{equation}
where $d$ is a positive constant. In particular, the first bound shows that exponential decay of the modified entropy implies exponential decay of the entropy.

The bound from below is easily established.
\begin{Lemma}\label{lowerbound}
Assume $b\leq \sqrt{ac}$. Then $\mathcal{E}[h]\geq k\,\D[h]$.
\end{Lemma}
\begin{proof}
By Young's inequality,  for all $\varepsilon>0$ we have
\[
g(\mathscr{A}_x h,\partial_p\h)\geq -\varepsilon g(\partial_ph,\partial_p\h)-\frac{1}{4\varepsilon}g(\mathscr{A}_x h,\mathscr{A}_x \h)\:,
\]
whence $2b\I_{xp}\geq -2b\varepsilon\I_{pp}-(b/2\varepsilon)\I_{xx}$ and so
\[
\mathcal{E}[h]\geq k\,\D[h]+\left(c-\frac{b}{2\varepsilon}\right)\I_{xx}[h]+(a-2\varepsilon b)\I_{pp}[h]\geq k\,\D[h]\:,
\]
provided $b/c\leq 2\varepsilon\leq a/b$.
\end{proof}
The bound from above, which requires the assumptions of the main theorem (except Assumption~\ref{logsobineqass}), is more complicated. Since 
\[
a\, \I_{pp}[h]+2b\,\I_{xp}[h]+c\,\I_{pp}[h]\leq\max (a+b,b+c)(\I_{pp}[h]+\I_{xx}[h])\:,
\]
it suffices to prove the following.
\begin{Proposition}\label{upperboundpro}
Let the Assumptions~\ref{cbcass},~\ref{gpos},~\ref{ABass},~\ref{Wass} hold. There exists a region $\Omega\subset\R^3$ such that $\Omega\subset\{(a,b,c):b\leq \sqrt{ac}\}$ and, for all $(a,b,c)\in\Omega$, there exists $d>0$ such that
\begin{equation}\label{upperbound}
\frac{d}{dt}\mathcal{E}[h]\leq -d\big(\I_{xx}[h]+\I_{pp}[h]\big)\:.
\end{equation}
\end{Proposition}
\begin{Remark}
The best constant in the inequality~\eqref{upperbound} may be written as $\bar{d}=\sup_{\bar{\Omega}} d$, where $\bar{\Omega}$ is the largest region for which Proposition~\ref{upperboundpro} holds. We refrain from computing it explicitly, since the method we use is anyway unsuitable to obtain the optimal rate of decay of the entropy.
\end{Remark}
The proof of the proposition is based on the following lemma.
\begin{Lemma}\label{bounds}
For all constants $\varepsilon_1,\dots,\varepsilon_{10}>0$ we have
\begin{equation}\label{dippbound}
\frac{d}{dt}\I_{pp}[h]\leq 2\varepsilon_1\I_{xx}[h]+\left(\frac{1}{2\varepsilon_1}-2\sigma_1\right)\I_{pp}[h]-2Q^2_{pp}\:,
\end{equation}
\begin{align}\label{dixpbound}
\frac{d}{dt}\I_{xp}[h]\leq &\Big[\varepsilon_2 \sigma+\varepsilon_3\sigma_1+2\varepsilon_5\gamma+\varepsilon_7\omega+\varepsilon_6\beta-1\Big]\I_{xx}[h]\nonumber\\
&+\frac{1}{4}\left(\frac{\sigma}{\varepsilon_2}+\frac{\sigma_1}{\varepsilon_3}+\frac{1}{\varepsilon_6}+\frac{1}{\varepsilon_7}\right)\I_{pp}[h]
+\left(2\varepsilon_4+\frac{1}{2\varepsilon_5}\right)Q^2_{pp}+\frac{1}{2\varepsilon_4}Q^2_{xp}\:,
\end{align}
\begin{align}\label{dixxbound}
\frac{d}{dt}\I_{xx}[h]\leq& \left(4\varepsilon_8\gamma+\frac{1}{2\varepsilon_9}+2\varepsilon_9\beta+2\varepsilon_{10}\omega+\frac{1}{2\varepsilon_{10}}\right)\I_{xx}[h]+\left(\frac{1}{\varepsilon_8}-2\right)Q^2_{xp}\:,
\end{align}
where
\[
Q^2_{pp}=\int h\,\partial_p^2\h\cdot\nabla_p^2\h\,dx\,d\mu\:,\quad Q_{xp}^2=\int h\,\partial_p(\mathscr{A}_x\h)\cdot\nabla_p(\mathscr{A}_x\h_*)\,dx\,d\mu\:.
\]
\end{Lemma}
\begin{proof}
The inequality~\eqref{dippbound} is a straightforward consequence of Lemma~\ref{dipp}, the inequality $\I_{xp}[h]\geq -\varepsilon_1\I_{xx}[h]-(4\varepsilon_1)^{-1}\I_{pp}[h]$, and Assumption~\ref{cbcass}.
We now prove~\eqref{dixpbound}. Using the identity
\begin{align*}
\widetilde{\Ric}(\mathscr{A}_x \h,\partial_p\h)=&\widetilde{\Ric}\left(\sqrt{\varepsilon_2}\mathscr{A}_x \h+\frac{1}{\sqrt{4\varepsilon_2}}\partial_p\h,\sqrt{\varepsilon_2}\mathscr{A}_x \h+\frac{1}{\sqrt{4\varepsilon_2}}\partial_p\h,\right)\\
&-\varepsilon_2\widetilde{\Ric}(\mathscr{A}_x \h,\mathscr{A}_x \h)-\frac{1}{4\varepsilon_2}\widetilde{\Ric}(\partial_p\h,\partial_p\h)\:,
\end{align*}
together with Assumption~\ref{cbcass} and $\I_{xp}[h]\geq -\varepsilon_3\I_{xx}[h]-(4\varepsilon_3)^{-1}\I_{pp}[h]$, we get
\begin{equation}\label{tem1}
-\int h\,\widetilde{\Ric}(\mathscr{A}_x \h,\partial_p\h)\leq (\varepsilon_2\sigma+\varepsilon_3\sigma_1)\I_{xx}[h]+\left(\frac{\sigma}{4\varepsilon_2}+\frac{\sigma_1}{4\varepsilon_3}\right)\I_{pp}[h]\:.
\end{equation}
By Young's inequality 
\begin{equation}
-2\int h\, \partial_p^2\h\cdot\nabla_p (\mathscr{A}_x \h)_*\leq 2\varepsilon_4\int h\,\partial_p^2\h\cdot\nabla_p^2\h+\frac{1}{2\varepsilon_4}\int h\,\partial_p(\mathscr{A}_x \h)\cdot\nabla_p(\mathscr{A}_x \h)_*\:.\label{tem2}
\end{equation}
By Young's inequality,~\eqref{ABobs}, Assumption~\ref{ABass} and~\eqref{gobs} 
\begin{align}\label{tem3}
2\int h\,\partial_p^2\h\cdot(\mathscr{C}_x\h)_*&\leq 2\varepsilon_5\!\int h\,\mathscr{C}_x\h\cdot(\mathscr{C}_x\h)_*+\frac{1}{2\varepsilon_5}\int h\,\partial^2_p\h\cdot\nabla^2_p\h\nonumber\\
&\leq 2\varepsilon_5\gamma\I_{xx}[h]+\frac{1}{2\varepsilon_5}\int h\,\partial^2_p\h\cdot\nabla^2_p\h\:.
\end{align}
Likewise
\begin{equation}\label{tem4}
\int g(\partial_p\h,\mathscr{B}_x h)\leq\varepsilon_6\int h\, g(\mathscr{B}_x \h,\mathscr{B}_x \h)+\frac{1}{4\varepsilon_6}\I_{pp}[h]\leq  \varepsilon_6\beta\I_{xx}[h]+\frac{1}{4\varepsilon_6}\I_{pp}[h]\:.
\end{equation}
Finally by Assumption~\ref{Wass},
\begin{align}\label{tem5}
\int\mathscr{C}_xh(W_*,\nabla_p\h)&=\int g(K^{(I)}\partial_{x^I}h,\partial_p\h)\nonumber\\
&\leq\varepsilon_7\int h\,g(K^{(I)},K^{(J)})\partial_{x^I}\h\,\partial_{x^J}\h+\frac{1}{4\varepsilon_7}\int g(\partial_ph,\partial_p\h)\nonumber\\
&\leq\varepsilon_7\omega\I_{xx}[h]+\frac{1}{4\varepsilon_7}\I_{pp}[h]\:.
\end{align}
Using the inequalities~\eqref{tem1}--\eqref{tem5} in Lemma~\ref{dixp} concludes the proof of~\eqref{dixpbound}. The proof of~\eqref{dixxbound} is similar. Reasoning as before one can prove that
\begin{align*}
&4\int h\,\partial_p(\mathscr{A}_x\h)\cdot(\mathscr{C}_xh)_*\leq 4\varepsilon_8\gamma\I_{xx}[h]+\frac{1}{\varepsilon_8}Q^2_{xp}\:,\\
&2\int g(\mathscr{A}_x\h,\mathscr{B}_xh)\leq \Big(2\varepsilon_9\beta+\frac{1}{2\varepsilon_9}\Big)\I_{xx}[h]\:,\\
&2\int\mathscr{C}_xh(W_*,(\mathscr{A}_x\h)_*)=2\int h\,g(K^{(I)}\partial_{x^I}\h,(\mathscr{A}_x\h))\leq\Big(2\varepsilon_{10}\omega+\frac{1}{2\varepsilon_{10}}\Big)\I_{xx}[h]
\end{align*}
and substituting in Lemma~\ref{dixx} completes the proof.
\end{proof}
\begin{Remark}
We are going to apply Lemma~\ref{bounds} for special values of the constants $\varepsilon_1,\dots,\varepsilon_{10}$. In its generality, Lemma~\ref{bounds} could be useful to improve the constant $d$ in~\eqref{upperbound}.
\end{Remark}
\begin{proof}[Proof of Proposition~\ref{upperboundpro}]
In the inequalities ~\eqref{dippbound}--\eqref{dixxbound} we set
\begin{align*}
&\varepsilon_1=(2a)^{-1}\:,\quad\varepsilon_2=\varepsilon_3=\varepsilon_6=\varepsilon_7=\frac{1}{4}(\sigma_2+\beta+\omega)^{-1}\:,\\
&\varepsilon_4=\frac{8}{7}(2+\beta+16\gamma+\omega)\:,\quad\varepsilon_5=(8\gamma)^{-1}\:,\quad\varepsilon_8=4\:,\ \varepsilon_9=\varepsilon_{10}=\frac{1}{2}\:.
\end{align*}
So doing we obtain
\begin{align*}
&\frac{d}{dt}\I_{pp}[h]\leq a^{-1}\I_{xx}[h]+(a-2\sigma_1)\I_{pp}[h]-2Q_{pp}^2\:,\\
&\frac{d}{dt}\I_{xp}[h]\leq -\frac{1}{2}\I_{xx}[h]+s_1s_2\I_{pp}[h]+\left(\frac{16s}{7}+4\gamma\right)Q_{pp}^2+\frac{7}{16s}Q_{xp}^2\:,\\
&\frac{d}{dt}\I_{xx}[h]\leq s\I_{xx}[h]-\frac{7}{4}Q_{xp}^2\:,
\end{align*}
where 
\[
s_1=\sigma_2+\beta+\omega\:,\ s_2=2+\sigma_2\:,\ s=2+\beta+16\gamma+\omega\:.
\]
Therefore
\begin{align*}
\frac{d}{dt}\mathcal{E}[h]&=k\,\frac{d}{dt}\D[h]+a\,\frac{d}{dt}\I_{pp}[h]+2b\,\frac{d}{dt}\I_{xp}[h]+c\,\frac{d}{dt}\I_{xx}[h]\\[0.1cm]
&\leq [-k+a(a-2\sigma_1)+2bs_1s_2]\,\I_{pp}[h]+(1+cs-b)\,\I_{xx}[h]\\
&\quad +2\left[b\left(\frac{16}{7}s+4\gamma\right)-a\right]Q_{pp}^2+\frac{7}{4}\left(\frac{b}{2s}-c\right)Q_{xp}^2\:.
\end{align*}
It is clear that the coefficient of $\I_{pp}[h]$ can be made negative by choosing $k$ sufficiently large, for all values of the other constants. To make the coefficients of $\I_{xx}[h]$, $Q_{pp}^2$, $Q_{xp}^2$ negative we require that
\[
b>1+cs\:,\quad b<\frac{a}{\frac{16}{7}s+4\gamma}\:,\quad b<2cs\:.
\]
This is possible as soon as
\[
a>(1+cs)(\frac{16}{7}s+4\gamma)\qquad\text{ and }\qquad c>s^{-1}\:.
\]
If we further require that $a>4s^2c$, then $2cs<\sqrt{ac}$ and therefore $b<2cs$ implies $b<\sqrt{ac}$ as well. This completes the proof of the proposition.
\end{proof}

\subsection{Completion of the proof}
To complete the proof of Theorem~\ref{maintheo} we appeal to Assumption~\ref{logsobineqass}.  Using the logarithmic Sobolev inequality~\eqref{logsobin} in~\eqref{upperbound} we obtain 
\[
\frac{d}{dt}\mathcal{E}[h]\leq -(2d\alpha)\D[h]
\] 
and combining with the second inequality in~\eqref{goal} we infer that there exists a constant $\lambda>0$ such that
\[
\frac{d}{dt}\mathcal{E}[h]\leq -\lambda\mathcal{E}[h]\:.
\]
Whence $\mathcal{E}[h]\leq \mathcal{E}[h_\mathrm{in}]\exp(-\lambda t)\leq C(\I_{xx}[h_\mathrm{in}]+\I_{pp}[h_\mathrm{in}])\exp(-\lambda t)$ and by the lower bound $\mathcal{E}[h]\geq k\,\D[h]$, see Lemma~\ref{lowerbound}, the entropy decays exponentially as well, which is the main claim of Theorem~\ref{maintheo}.

\begin{appendix}
\section{Appendix: Cauchy's problem}
In this appendix we discuss the global existence and uniqueness of solutions to the Cauchy problem for equation~\eqref{maineq} in the case when the dimensions of the spaces $\N$ and $\M$ coincide, i.e., $N=M$. 
The following discussion is based on the methods introduced in~\cite[Ch.~5]{helffer}, except that we work in a different functions space. To adhere with the conventions used in~\cite{helffer}, we rewrite~\eqref{maineq} as
\[
\partial_th+Ah=0\:,
\]
where
\[
A=-\Delta_p-W+v(p)\cdot\nabla_x=-L+T\:.
\]
The domain $D(A)$ of the operator $A$ is chosen as the space of $C^\infty$ functions on $\mathbb{T}^N\times\R^N$ with compact support in the $p\in\R^N$ variable, which is dense in $\mathcal{H}:=L^2(dxd\mu)$. Our first purpose is to prove that the closure of $A$ generates a contraction (dissipative) semigroup in $\mathcal{H}$. To this end we need to assume that the quantities $g,v,E$ are $C^\infty$.
Furthermore we assume that 
\begin{equation}\label{boundstemp}
\frac{g^{ij}(p)}{|p|^2}\to 0\:,\ \text{as }|p|\to\infty\ \ \forall\, i,j=1,\dots N\:.
\end{equation}
We divide the proof in three steps.
\subsubsection*{Step 1: $A$ is accretive.} By~\eqref{IPF1},
\[
<\!h|Ah\!>_\mathcal{H}=-<\!h|Lh\!>_\mathcal{H}+<\!h|Th\!>_\mathcal{H}=\int g(\partial_ph,\partial_ph)\geq 0\:.
\]
Recall that all integrals are extended over $\mathbb{T}^N\times\R^N$ with measure $dxd\mu$. 
\subsubsection*{Step 2: $A$ is hypoelliptic.}
Let $a=\sqrt{g^{-1}}$ (i.e., the positive definite matrix such that $a^2=g^{-1}$). A straightforward calculation shows that the operator $-A$ can be written in H\"ormander's form:
\[
-A=\sum_{i=1}^N Y_{(i)}^2+Y_0\:,
\] 
where 
\begin{align*}
&Y_0h=(\Div a)\cdot a\nabla_ph-g(\partial_pE,\partial_ph)-Th\:,\\
&Y_{(i)}h=a^{k}_{i}\partial_{p^k}h\:.
\end{align*}
To prove that the operator $A$ is hypoelliptic, we will show that $-A$ satisfies a rank 2 Hormander's condition, namely that the vector fields
\[
Y_{(i)}\:,\quad Z_{(i)}:=[Y_0,Y_{(i)}]
\]
form a basis of $\R^{2N}$. To this purpose we observe that
\[
Z_{(i)}=B_i^k\partial_{p^k}+C_i^I\partial_{x^I}=B^k_iP_{(k)}+C_i^IX_{(I)}\:,
\]
where 
\[
C_i^I=a^k_ i\partial_{p^k}v^{(I)}
\] 
and $B$ is a $p-$dependent $N\times N$ matrix, whose exact form is irrelevant for what follows. Thus the linear transformation $\{X_{(I)},P_{(i)}\}\to\{Y_{(i)},Z_{(i)}\}$ is represented by the matrix 
\[
F=\left(\begin{array}{cc} 0 & a\\ C & B\end{array}\right).
\]
The determinant of $F$ is given by
\[
|\det F\,|=\det a|\det C\,|=\det g|\det (\partial_{p^k}v^{(I)})\,|\:,
\]
which is positive because $\det g>0$ and, by Assumption~\ref{gpos}, the determinant of the matrix $\partial_{p^k}v^{(I)}$ is non-zero. Thus $\{Y_{(i)},Z_{(i)}\}$ is a new basis of $\R^{2N}$, concluding the proof.

\subsubsection*{Step 3: The closure of $A$ is maximally accretive.}
By~\cite[Th. 5.4]{helffer} (see also~\cite{LP}), it is enough to prove that the range of $\lambda+A$ is dense in $\mathcal{H}$, for some $\lambda>0$. We need to show that if $h\in\mathcal{H}$ is such that
\begin{equation}\label{maxac}
<\!h|(\lambda+A)f\!>_\mathcal{H}=0\:,\quad \text{for all } f\in D(A)\:,
\end{equation}
then $h=0$. Note that~\eqref{maxac} implies that $h$ is a distributional solution of $$(\lambda-L-T)h=0\:.$$ Since the operator in the left hand side of the latter equation is hypoelliptic, then we may assume that $h\in C^\infty$. Let us begin by proving that the following identity holds:
\begin{equation}\label{impid}
\lambda\int\phi^2h^2+\int g(\partial_p(\phi h),\partial_p(\phi h))=\int h^2g(\partial_p\phi,\partial_p\phi)-\int h^2\phi T\phi\:,
\end{equation}
for all $\phi\in D(A)$.
To prove~\eqref{impid}, we use that, by (i)-(ii) of Lemma~\ref{identities},
\[
(\lambda+A)(f_1f_2)=f_1(\lambda+A)f_2+f_2Af_1-2g(\partial_pf_1,\partial_pf_2)\:,\quad\text{for all }f_1,f_2\in C^\infty\:.
\]
Setting $f_1=\phi$, $f_2=\phi h$ and multiplying by $h$ the resulting identity we get
\[
\phi h(\lambda+A)(\phi h)=h(\lambda+A)(\phi^2 h)-h^2\phi A\phi+2hg(\partial_p\phi,\partial_p(\phi h))\:.
\]
Integrating and using that $<\!h|(\lambda+A)(\phi^2 h)\!>_\mathcal{H}=0$, by~\eqref{maxac}, we have
\[
\int \phi h (\lambda+A)(\phi h)=\int h^2\phi L\phi-\int h^2\phi T\phi+2\int h g(\partial_p\phi,\partial_p(\phi h))\:.
\]
Using~\eqref{IPF1} in the l.h.s. and in the first term in the r.h.s. of the previous identity completes the proof of~\eqref{impid}. 
Now let $k\in\mathbb{N}$ and choose a family of test functions $\phi_k$ of the form
\[
\phi_k(x,p)=\psi(p/k)\:, 
\]
where $\psi\in C^\infty_c$, $0\leq \psi\leq 1$, $\psi=1$ for $p\in B(0,1/2)$ and $\mathrm{supp\,\psi}\subset B(0,1)$. Whence $T\phi_k=0$. Substituting in~\eqref{impid} we obtain
\[
\lambda\int \phi_k^2 h^2\leq \frac{1}{k^2}\int h^2 g^{ij}\partial_{p^i}\psi\partial_{p^j}\psi\,\chi_{|p|\leq k}\:.
\]
Having assumed~\eqref{boundstemp}, we obtain
\[
\lambda\int\phi^2_k h^2\leq  \epsilon(k)\:,
\]
where $\epsilon(k)\to 0$ as $k\to\infty$. 
This finally entails that $h\equiv 0$.

We may now sketch the proof of the global well-posedeness of the Cauchy problem in the setting of $L^1$ solutions.

\begin{Theorem}\label{globalex}
Let $g,v,E$ be $C^\infty$ functions such that~\eqref{boundstemp} holds. 
Given $0\leq h_\mathrm{in}\in L^1(\mathbb{T}^N\times\R^N; dxd\mu)$, there exists a unique $h\in C([0,\infty); L^1(\mathbb{T}^N\times\R^N;dxd\mu))$ solution of~\eqref{kineticFP3} with initial datum $h_\mathrm{in}$. 
\end{Theorem}
\begin{proof}
Approximate the initial datum by a sequence $h_{\mathrm{in},m}$ of smooth, non-negative functions in the domain $D(A)$. By the preceding result, for each fixed $m\in\mathbb{N}$ there exists a unique  $h_m\in C([0,\infty),L^2(dx\,d\mu))$, solution of~\eqref{kineticFP3}. Moreover by standard methods (see~\cite{CLR, Va} for instance) one can prove the $L^1$-contraction property: $\|h_k-h_m\|_{L^1}\leq\|h_{\mathrm{in},k}-h_{\mathrm{in},m}\|_{L^1}$. Thus the sequence $h_m$ converges in $L^1$ to a solution with the regularity stated in the theorem. The uniqueness is also a consequence of the $L^1$-contraction property.   The non-negativity of  solutions  can be proved by studying the evolution of a suitable regularization of $\mathrm{sign} (h)$ (see again~\cite{CLR,Va}). 
\end{proof}

\section{Appendix: Validity of the Logarithmic Sobolev inequality}\label{logproof}
In this appendix we provide a sufficient condition for the validity of the logarithmic Sobolev inequality~\eqref{logsobin}. 

Let $A_{IJ}$ denote the matrix inverse of $A^{IJ}=g(\partial_pv^{(I)},\partial_pv^{(J)})$, i.e., $A^{IJ}A_{JK}=\delta^I_K$. We define the Riemannian metric $G$ on $\M\times\N$ as
\begin{equation}\label{metricG}
G=g_{ij}dp^i\otimes dp^j+A_{IJ}dx^I\otimes dx^J\:.
\end{equation}
Moreover we define the vector field $Q\in\X(\M)$ as
\begin{equation}\label{vectorQ}
Q=W-\partial_{p}\log\sqrt{\det({A_{IJ}})}\:.
\end{equation}
\begin{Theorem}\label{sufflogsob}
The inequality~\eqref{logsobin} holds if  
\begin{equation}\label{complicatedass}
\Ric^G(Z,Z)-(\nabla^GQ_*)(Z,Z)\geq \alpha\,G(Z,Z)\:,\quad\text{for all }Z\in\X(\M\times\N)\:,
\end{equation}
where $\Ric^G$ is the Ricci tensor of $G$ and $\nabla^G$ is the covariant differential associated with $G$. 
\end{Theorem}
\begin{proof}
Consider the non-degenerate Fokker-Planck equation 
\begin{equation}\label{auxeq}
\partial_t f=\Delta^Gf+Qf
\end{equation}
on $\M\times\N$. The entropy functional and entropy dissipation functional associated with~\eqref{auxeq} are exactly $\D$ and $\I=\I_{pp}+\I_{xx}$. Moreover~\eqref{complicatedass} asserts that the metric $G$ and the vector field $Q$ verify the curvature Bakry-Emery bound condition. Thus the logarithmic Sobolev inequality~\eqref{logsobin} follows by the results in~\cite{BE}.
\end{proof}

Using the relations between $\Ric^G,\nabla^G$ and $\Ric,\nabla_p$,  the bound~\eqref{complicatedass} can be expressed in terms of inequalities on the quantities $g, W$. These inequalities are in general very complicated, unless $A_{IJ}$ enjoys some simple structure, as in the statement of Corollary~\ref{corollo} below.

\begin{Corollary}\label{corollo}
Let $A_{IJ}$---the matrix inverse of $A^{IJ}=g(\partial_p v^{(I)},\partial_p v^{(J)})$---be of the form
\begin{equation}\label{zeta}
A_{IJ}=\zeta (p)^2\delta_{IJ}\:,
\end{equation}
for some smooth function $\zeta:\R^M\to (0,\infty)$. If there exist two constants $\kappa_1>\kappa_2\geq 0$ such that
\begin{align*}
&\left[\Ric-\nabla^2_p\log u-\frac{N}{\zeta^2} \partial_p\zeta\otimes\partial_p\zeta\right](X,X)\geq \kappa_1g(X,X)\:,\\
&\left[\Delta_p\log\zeta+\frac{g(\partial_pu,\partial_p\zeta)}{\zeta u}\right]\leq \kappa_2\:,
\end{align*}
for all $X,p\in\R^M$, then~\eqref{complicatedass}  holds.
\end{Corollary}
\begin{proof}
For the proof we observe that when $A_{IJ}$ has the form~\eqref{zeta}, the Riemannian manifold $(\M\times\N,G)$ is the warped product of the manifolds $(\R^M,g)$ and $(\mathbb{T}^N,\delta)$, where $\delta$ is the flat Euclidean metric on the torus. 
See~\cite[Ch.~7]{oneill} for an introduction to the geometry of warped product manifolds. In particular, by Corollary 43 of~\cite[Ch.~7]{oneill} we have that, for all horizontal (i.e., tangent to $\M$) vector fields $X,Y$ and vertical (i.e. tangent to $\N$) vector fields $V,W$, the following identities hold:
\begin{align*}
&\Ric^G(X,Y)=\Ric(X,Y)-\frac{N}{\zeta}\nabla^2_p\zeta(X,Y)\:,\\
&\Ric^G(X,V)=0\:,\\
&\Ric^G(V,W)=-\left(\frac{\Delta_p\zeta}{\zeta}+(N-1)\frac{g(\partial_p\zeta,\partial_p\zeta)}{\zeta^2}\right)G(V,W)\:.
\end{align*}  
Moreover the vector field~\eqref{vectorQ} becomes
\[
Q=\partial_p(\log u-N\log\zeta)
\]
and by Proposition 35 of~\cite[Ch.~7]{oneill} we have
\begin{align*}
&\nabla^GQ_*(X,Y)=\nabla_pQ_*(X,Y)=\left(\nabla^2_p\log u-N\nabla_p^2\log\zeta\right)(X,Y)\:,\\
&\nabla^GQ_*(X,V)=0\:,\\
&\nabla^GQ_*(V,W)=\frac{Q(\zeta)}{\zeta}G(V,W)=\left(\frac{g(\partial_pu,\partial_p\zeta)}{\zeta u}-N\frac{g(\partial_p\zeta,\partial_p\zeta)}{\zeta^2}\right)G(V,W)\:.
\end{align*}
Now let $X=Y=Z_{|_\mathcal{M}}$, the projection of $Z$ onto $\mathcal{M}$, and $V=W=Z_{|_\mathcal{N}}$, the projection of $Z$ onto $\mathcal{N}$. Writing $Z=Z_{|_\mathcal{M}}+Z_{|_\mathcal{N}}$ we have
\begin{align*}
(\Ric^G-\nabla^GQ_*)(Z,Z)=&\left(\Ric-\frac{N}{\zeta}\nabla^2_p\zeta\right)(Z_{|_\mathcal{M}},Z_{|_\mathcal{M}})\\
&-\left(\frac{\Delta_p\zeta}{\zeta}+(N-1)\frac{g(\partial_p\zeta,\partial_p\zeta)}{\zeta^2}\right)G(Z_{|_\mathcal{N}},Z_{|_\mathcal{N}})\\
&-\left(\nabla^2_p\log u-N\nabla_p^2\log\zeta\right)(Z_{|_\mathcal{M}},Z_{|_\mathcal{M}})\\
&-\left(\frac{g(\partial_pu,\partial_p\zeta)}{\zeta u}-N\frac{g(\partial_p\zeta,\partial_p\zeta)}{\zeta^2}\right)G(Z_{|_\mathcal{N}},Z_{|_\mathcal{N}})\\
=&\left(\Ric-\nabla^2_p\log u-\frac{N}{\zeta^2}\partial_p\zeta\otimes\partial_p\zeta\right)(Z_{|_\mathcal{M}},Z_{|_\mathcal{M}})\\
&-\left(\Delta_p\log\zeta+\frac{g(\partial_pu,\partial_p\zeta)}{\zeta u}\right)G(Z_{|_\mathcal{N}},Z_{|_\mathcal{N}})\:.
\end{align*}
Under the given assumptions we have
\begin{align*}
(\Ric^G-\nabla^GQ_*)(Z,Z)&\geq \kappa_1G(Z,Z)-\kappa_2G(Z_{|_\mathcal{N}},Z_{|_\mathcal{N}})\\
&=(\kappa_1-\kappa_2)G(Z,Z)+\kappa_2G(Z_{|_\mathcal{M}},Z_{|_\mathcal{M}})\geq(\kappa_1-\kappa_2)G(Z,Z)
\end{align*}
and the conclusion of Corollary~\ref{corollo} follows.

\end{proof}

An example that is covered by Corollary~\ref{corollo} is the Fokker-Planck equation~\eqref{kineticFP2} with the classical velocity field $v(p)=p$ and an isotropic diffusion matrix, i.e.
\[
D_{ij}(p)=\Pi(p)^2\delta_{ij}\:,
\]
where $\Pi$ is a positive function. In particular $g^{ij}=\Pi(p)^2\delta^{ij}$.
Thus, since $v^{(I)}=p^I$ (all indexes run from 1 to $N$ in this example), we have
\[
A^{IJ}=g^{ij}\partial_{p^i}v^{(I)}\partial_{p^j}v^{(J)}=g^{ij}\delta_i^I\delta_j^J=g^{IJ}=\Pi(p)^2\delta^{IJ}\:,
\]
and so~\eqref{zeta} holds with $\zeta(p)=1/\Pi(p)$.

\end{appendix}

\vspace{0.5cm}

\noindent {\bf Acknowledgments:} This work was completed while the author was a long term participant at the program ``Partial Differential Equations in Kinetic Theories" at the Isaac Newton Institute in Cambridge (UK). The valuable comments and suggestions of the two anonymous referees are also acknowledged.

\end{document}